\documentclass[10pt]{amsart}

\usepackage{amssymb,amsmath,euscript}

\usepackage{}

\newtheorem{theorem}{Theorem}[section]
\newtheorem{lemma}[theorem]{Lemma}
\newtheorem{proposition}[theorem]{Proposition}
\newtheorem{conjecture}[theorem]{Conjecture}
\newtheorem{corollary}[theorem]{Corollary}

\theoremstyle{definition}
\newtheorem{definition}[theorem]{Definition}

\theoremstyle{remark}
\newtheorem{remark}[theorem]{Remark}
\def\Fq{{\mathbb F}_q}
\def\a{{\alpha}}

\def\FF{{\mathbb F}}
\def\PP{{\mathbb P}}

\newcommand{\codim}{\operatorname{codim}}
\newcommand{\wt}{\operatorname{wt}}

\newcommand{\GL}{\mathrm{GL}}

\def\PP{{\mathbb P}}

\def\O{{\Omega_{\alpha}(\ell,m)}}
\def\Op{{\Omega_{\alpha'}(\ell -1,m)}}
\def\C{{C_{\alpha}(\ell, m)}}
\def\Cp{{C_{\alpha'}(\ell-1, m)}}
\def\Q{{q^{\delta(\alpha)}}}
\def\Qp{{q^{\delta(\alpha')}}}

\def\a{{\alpha}}

\begin{document}
	
	\title[Minimum Weight Codewords of Schubert Codes]{Minimum Distance and the Minimum Weight Codewords of Schubert Codes}
	
	
	\author{Sudhir R. Ghorpade}
	\address{Department of Mathematics, 
	Indian Institute of Technology Bombay,\newline \indent
	Powai, Mumbai 400076, India.}
	\email{srg@math.iitb.ac.in}
	\thanks{The first named author was partially supported by Indo-Russian project INT/RFBR/P-114 from the Department of Science \& Technology, Govt. of India and  IRCC Award grant 12IRAWD009 from IIT Bombay.}
	\author{Prasant Singh}
	\address{Department of Mathematics,
	Indian Institute of Technology Bombay,\newline \indent
	Powai, Mumbai 400076, India.}
	\thanks{The second named author was partially supported by a doctoral fellowship from the Council of Scientific and Industrial Research, India.}
	\email{psinghprasant@gmail.com}
	

	\date{\today}
	
\begin{abstract}
We consider linear codes associated to Schubert varieties in Grassmannians.  A formula for the minimum distance of these codes was conjectured in 2000 and after having been established 
in various special cases, it was proved in 2008 by Xiang. We give an alternative proof of this formula. Further, we propose a characterization of the minimum weight codewords of Schubert codes by introducing the notion of Schubert decomposable elements of certain exterior powers. It is shown that codewords corresponding to Schubert decomposable elements are of minimum weight and also that the converse is true in many cases. A lower bound, and in some cases, an exact formula, for the number of minimum weight codewords of Schubert codes is also given. From a geometric point of view, these results correspond to determining the maximum number of $\Fq$-rational points that can lie on a hyperplane section of a Schubert variety in a Grassmannian with its nondegenerate embedding in a projective subspace of the Pl\"ucker projective space, and also the number of hyperplanes for which the maximum is attained. 
\end{abstract}

\maketitle

\section{Introduction}

Fix a prime power $q$ and positive integers $\ell,m$ with $\ell \le m$. 
Let $\mathbb{F}_q$ denote the finite field with $q$ elements 
and let $V$ be a vector space over $\mathbb{F}_q$ of dimension $m$. 
To the Grassmannian $G_{\ell,m}$ of all $\ell$-dimensional linear subspaces of $V$, one can associate in a natural way an $[n,k]_q$-code, i.e., a $q$-ary linear code of length $n$ and dimension $k$, where 
\begin{equation}
\label{eq:1}
n = { {m} \brack {\ell}}_q :=
 \frac{ (q^m -1)(q^{m} - q) \cdots (q^{m} - q^{\ell -1})}
{(q^{\ell} -1)(q^{\ell} - q) \cdots (q^{\ell} - q^{\ell -1})} \quad
{\rm and} \quad k = { {m} \choose {\ell}}.
\end{equation}
This code is denoted by $C(\ell,m)$ and is called \emph{Grassmann code}. The study of Grassmann codes goes back to the work of Ryan \cite{CR1, CR2,CKR} in the late 1980's and was continued by Nogin \cite{N} and several authors (see, e.g., \cite{GL, GPP, HJR2, GK, BP} and the references therein). It is now known that Grassmann codes possess a number of interesting properties. For instance, their minimum weight is known and is given by the following beautiful formula of Nogin \cite{N}:
\begin{equation}
\label{eq:2}
d\left( C(\ell,m) \right) = q^{\delta} \quad \text{where} \quad \delta:= \ell (m-\ell).
\end{equation}
Furthermore, several generalized Hamming weights are known, the automorphism group has been determined and is known to be fairly large, the duals of Grassmann codes have a very low minimum distance (namely, $3$) 
and 
the minimum weight codewords of $C(\ell,m)^{\perp}$ generate $C(\ell,m)^{\perp}$. In fact, as the results of \cite{BP, FP} show, Grassmann codes can be regarded as regular LDPC codes and also as a Tanner codes with a small component code, namely, $C(1,2)$. 

Schubert codes are a natural generalization of Grassmann codes and were introduced in \cite{GL} around the turn of the last century. These are linear codes $\C$ associated to Schubert subvarieties $\O$ of the 
Grassmannian $G_{\ell,m}$ and are indexed by $\ell$-tuples
 $\alpha=({\a}_1,\ldots,{\a}_\ell)$ of positive integers with $1 \le \a_1 < \cdots < \a_{\ell} \le m$. The Grassmann codes are a special case where 
 $\a_i = m-\ell+i$ for $i=1, \dots, \ell$. It was shown in \cite{GL} that the minimum distance of $\C$ satisfies the inequality
 \begin{equation}
\label{eq:3}
 d\left( C_{\a}(\ell,m) \right) \le q^{\delta(\a)} \quad \text{where} \quad \delta(\a):= \sum_{i=1}^{\ell} (\a_i - i).
\end{equation}
Further, it was conjectured in \cite{GL} that the inequality in \eqref{eq:3} is, in fact, an equality. We will refer to this conjecture as the Minimum Distance Conjecture, or in short, the MDC. 
When $\a_i = m-\ell+i$ for $i=1, \dots, \ell$, we have $\delta(\a) = \delta$ and so the MDC holds, thanks to 
\eqref{eq:2}.  In the case $\ell=2$, the MDC was proved in the affirmative by Chen \cite{HC} and, independently, by Guerra and Vincenti \cite{GV}. 
An explicit formula for the length $n_{\a}$ and dimension $k_{\a}$ of $\C$ in the case $\ell=2$ was also given in \cite{HC}, 
while \cite{GV} gave a general, even if complicated, formula for $n_{\a}$ for arbitrary $\ell$. 
Later, in \cite{GT}, the MDC was established for Schubert divisors (i.e., in the case $\delta(\a) = \delta -1$) and general formulas for $n_{\a}$ and  $k_{\a}$ were obtained, namely, 
 \begin{equation}
\label{eq:4}
n_{\a} = \sum_{\beta \le \a} q^{\delta(\beta)} \quad \text{and} \quad k_{\a} = \det_{1\le i,j\le \ell} \left( { {\a_j - j +1} \choose {i - j + 1}}\right),
\end{equation}
where the above 
summation is over all $\ell$-tuples $\beta = (\beta_1, \dots , \beta_{\ell})$ of integers satisfying 
$1\le \beta_1 < \dots < \beta_{\ell} \le m$ and $\beta_i \le \alpha_i$ for $i=1, \dots , \ell$ and 
$\delta(\beta) := \sum_{i=1}^{\ell} (\beta_i - i)$. 
An affirmative answer to the MDC was eventually proposed by Xiang \cite{X}, where an alternative proof of the inequality in \eqref{eq:3} is given and a clever and rather involved proof of the other inequality is also given. While one doesn't doubt the veracity of Xiang's proof, it has been felt that a cleaner and more transparent proof of the MDC would be desirable. 
With this in view, we give in this paper an alternative coordinate-free argument to establish the MDC in the general case. Further, we take up the problem of characterizing the minimum weight codewords of Schubert codes and determining their number. In the case of Grassmann codes, there is a nice characterization that was given already by Nogin \cite{N}. To explain this, let us note that the codewords of $C(\ell,m)$ are indexed by elements $f$ of the exterior power $\bigwedge^{\ell}V$ and may be denoted by $c_f$. In fact, $c_f = \left(f\wedge P_1, \dots , f\wedge P_n\right)$, where $P_1, \dots, P_n$ is a fixed set of representatives in $\bigwedge^{m-\ell}V$ of  (the $\Fq$-rational points of) $G_{\ell, m}$. 
The map $\bigwedge^{m-\ell}V \to C(\ell,m)$ given by $f\mapsto c_f$ is a linear bijection. The characterization is simply that $c_f$ is a minimum weight codeword 
 $C(\ell,m)$ if and only if $f$ is decomposable, i.e.,  $f=f_1\wedge \ldots \wedge f_{m-\ell}$ for some linearly independent 
 $f_1, \ldots , f_{m-\ell} \in V$. The Schubert code $\C$ can be viewed as a puncturing of the Grassmann code $C(\ell,m)$ at the complement of $\O$ in $G_{\ell,m}$.  
The codewords of $\C$ can still be indexed by $f\in \bigwedge^{m-\ell}V$ and are given by 
$\left(f\wedge P_1, \dots , f\wedge P_{n_{\a}}\right)$, where $P_1, \dots, P_{n_{\a}}$ is a fixed set of representatives in $\bigwedge^{m-\ell}V$ of  (the $\Fq$-rational points of) $\O$; we will continue to denote these by $c_f$.  
 However, in general, the map $f\mapsto c_f$ of $\bigwedge^{m-\ell}V \to \C$ is surjective, but not injective. 
 This makes the case of Schubert codes more difficult 
 and a straightforward generalization of the characterization of minimum weight codewords of Grassmann codes does not hold for Schubert codes. 
 It turns out that one needs here a stronger and more subtle notion of decomposability that we call \emph{Schubert decomposability}. We propose a new conjecture that the minimum weight codewords of $\C$ are precisely those that correspond to Schubert decomposable elements of  $\bigwedge^{m-\ell}V$. We prove several aspects of this conjecture. Thus we show that codewords indexed by Schubert decomposable elements of $\bigwedge^{m-\ell}V$ are minimum weight codewords of $\C$ and we also show that if $c_f \in \C$ is a minimum weight codeword for some decomposable $f\in \bigwedge^{m-\ell}V$, then $f$ is Schubert decomposable. What remains to be seen is whether every minimum weight codeword of $\C$ can be indexed by a decomposable element of $\bigwedge^{m-\ell}V$.  We show that this is indeed the case when $\ell=2$ or when $\alpha$ is ``completely non-consecutive'', i.e., when ${\a}_i-{\a}_{i-1}\geq 2$ for all $i=1, \dots , \ell$. Thus, the new conjecture is established in these cases. 
 We also 
 give an explicit lower bound for the number of minimum weight codewords of $\C$, and observe that it gives the exact value if 
 our new conjecture is true. Of course in the completely non-consecutive case or when $\ell=2$, this becomes an unconditional result. 
We also show that unlike Grassmann codes, the minimum weight codewords of a Schubert code do not, in general, generate the code. On the other hand, one knows from the recent work of Pi\~nero \cite{FP} that the duals of Schubert codes have the same low minimum distance as that of $C(\ell,m)^{\perp}$ and moreover, the minimum 
weight codewords of  $\C^{\perp}$ generate $\C^{\perp}$. 
 
 The results of this paper have a 
 geometric interpretation that may be of independent interest. Indeed, $\O$ admits 
 a nondegenerate embedding in $\PP^{k_\a -1}$ and using 
 the language of projective systems (see, e.g., \cite[\S 1.1]{TVN}), we see that determining the minimum distance $d(\C)$ is equivalent to determining the maximum number of $\Fq$-rational points in sections of $\O$ by hyperplanes in~$\PP^{k_\a -1}$  since 
$$
m_\a :=  \max\{\left|\left(\O\cap H \right) (\Fq)\right| : H  \text{ hyperplane in }\PP^{k_\a -1} \} = n_\a - d(\C).
$$
Furthermore, if $M_{\a}$ is the number of minimum weight codewords of $\C$, then
$$
M_{\a} = \left| \left\{ H : H  \text{ hyperplane in }\PP^{k_\a -1} \text{ with } \left|\left(\O\cap H \right) (\Fq)\right|  =  m_\a \right\} \right|.
$$
 

\section{Preliminaries}
\label{sec:prelim}

In this section, we recall some basic notions and set the notations and terminology used in the rest of this paper. 
As in the Introduction, a prime power $q$ and 
integers $\ell,m$ with $1\le \ell \le m$ will be kept fixed throughout this paper. We have frequently used the notation $A:=B$ to mean that $A$ is defined to be equal to $B$. 

\subsection{Linear Codes} 
Let $n,k$ be positive integers. 
By an \emph{$[n,k]_q$-code} we mean a linear $k$-dimensional subspace of $\Fq^n$. Let $C$ be an  $[n,k]_q$-code. The parameters $n$ and $k$ are called the \emph{length} and the \emph{dimension} of $C$, respectively, whereas elements of $C$ are 
usually referred to as \emph{codewords}. Given a codeword $c=(c_1, \dots , c_n)$ of $C$, the \emph{Hamming weight} of $c$ will be denoted by $\wt(c)$; this is simply the number of $i\in \{1, \dots , n\}$ for which $c_i\ne 0$. The \emph{minimum distance} of $C$ is denoted by $d(C)$ and can be defined as $\min\{\wt(c) : c\in C \text{ with } c\ne 0\}$. Elements $c\in C$ satisfying $\wt(c) = d(C)$ are called the \emph{minimum weight codewords} of $C$. 
 
\subsection{Grassmann and Schubert Varieties}
Let $\FF$ be a field (later we will mainly take $\FF = \Fq$, but for now it can be an arbitrary field) and $V$ be an $m$-dimensional vector space over $\FF$. For a nonnegative integer $d$, we let $\bigwedge^dV$ denote the $d$th exterior power of $V$; this is a vector space over $\FF$ of dimension ${m\choose d}$. Fixing a basis of $V$, we can (and will) identify $\bigwedge^mV$ with $\FF$. Also the dual $(\bigwedge^dV)^*$ is canonically identified with $\bigwedge^{m-d}V$. An element $f$ of $\bigwedge^dV$ is said to be \emph{decomposable} if $f\ne 0$ and $f= f_1\wedge \cdots \wedge f_d$ for some $f_1, \dots , f_d\in V$. In general, elements of $\bigwedge^dV$ are $\FF$-linear combinations of decomposable elements. The \emph{annihilator} of any $f\in \bigwedge^dV$ is the subspace of $V$ denoted by $V_f$ and defined by 
$$
V_f:= \{ x\in V: f\wedge x =0\}. 
$$
Evidently, $f=0$ if and only if 
$\dim V_f =m$. Now suppose $d<m$. Then the following characterization is well-known; see, e.g., \cite[Thm. 1.1]{M}:
 \begin{equation}
\label{eq:5}
f \text{ is decomposable}  \Longleftrightarrow \dim V_f =d.
\end{equation}
Note that if $f$ is decomposable and $f=f_1\wedge \dots \wedge f_d$, then $\{f_1, \dots , f_d\}$ is a 
basis of $V_f$. And if $\{g_1, \dots , g_d\}$ is an arbitrary basis of $V_f$, then $f = \lambda (g_1\wedge \dots \wedge g_d)$ where $\lambda\in \FF$ is a nonzero scalar given by the determinant of the change-of-basis matrix. 

As in the Introduction, the Grassmannian $G_{\ell,m} = G_{\ell}(V)$ may be defined by 
$$
G_{\ell,m}:= \{ L : L \text{ is a $\ell$-dimensional subspace of } V \}.
$$
Elements of $G_{\ell,m}$ can be identified with the points of the projective space $\PP\big(\bigwedge^{\ell} V \big)$ via the \emph{Pl\"ucker embedding}, which associates to a subspace $L\in G_{\ell,m}$ with $\FF$-basis $\{v_1, \dots , v_{\ell}\}$ the class $[v_1\wedge\ldots\wedge v_\ell]$ of $v_1\wedge\ldots\wedge v_\ell \in \bigwedge^{\ell} V $. It is well-known that this is a well-defined embedding under which $G_{\ell,m}$ corresponds to a projective algebraic variety in $\PP\big(\bigwedge^{\ell} V \big)$ defined by the vanishing of certain quadratic homogeneous polynomials with integer coefficients. Moreover, the embedding is nondegenerate, i.e., $G_{\ell,m}$ is not contained in any hyperplane of $\PP\big(\bigwedge^{\ell} V \big)$. One can also view $G_{\ell,m}$ as a quotient of $\GL_m(\FF)$. Indeed, 
the group  $\GL(V)$ of invertible linear maps of $V \to V$ acts transitively on $G_{\ell,m}$ and so $G_{\ell,m}$ can be viewed as the homogeneous space $\GL(V)/P_{\ell}$, where 
$P_{\ell}$ is the parabolic subgroup given by the stabilizer of 
a fixed $\ell$-dimensional subspace of $V$. As this indicates, $G_{\ell,m}$ is a nonsingular variety of dimension $\delta:= \ell (m-\ell)$. 
When $\ell =m$,  the Grassmannian is a particularly simple object, namely the singleton set $\{V\}$, or the projective space $\PP^0$ consisting of a single point. 
Thus, to avoid trivialities, we shall henceforth assume that $1 \le \ell < m$. 

Now let us fix a partial flag $A_1\subset \dots \subset A_\ell$ of nonzero subspaces of $V$. Let ${\a}_i:= \dim A_i$ for $i=1, \dots , \ell$. We sometimes refer to  $\alpha:= ({\a}_1,\ldots,{\a}_\ell)$ as the \emph{dimension sequence} of the partial flag $A_1\subset \dots \subset A_\ell$. 
Note that $1 \le \a_1 < \cdots < \a_{\ell} \le m$. The \emph{Schubert variety} corresponding to this partial flag depends essentially on the dimension sequence  $\alpha$ 
and is defined by 
 \begin{equation}
\label{eq:6}
\O := \{L\in G_{\ell,m}:\dim (L\cap A_i)\geq i \text{ for all } i=1, \dots , \ell \}.
\end{equation}
With respect to the Pl\"ucker embedding of $G_{\ell,m}$, the Schubert variety $\O$ corresponds to a subset of $\PP\big(\bigwedge^{\ell} V \big)$ given by the intersection of $G_{\ell,m}$ with a bunch of Pl\"ucker  coordinate hyperplanes. As such $\O$ is indeed a projective variety that is known to be nondegenerately embedded in $\PP^{k_{\a}-1}$, where $k_{\a}$ is as in \eqref{eq:4}. Note 
that the elements of $\O$ are precisely those $L\in G_{\ell,m}$ for which there is a basis of the form $\{v_1, \dots , v_{\ell}\}$ with the property $v_i \in A_i$ for $i=1, \dots , \ell$. Thus, 
$$
\O = \{[v_1\wedge\ldots\wedge v_\ell]: v_1, \dots , v_{\ell}\in V \text{ linearly independent and } v_i\in A_i \;  \forall \, i\}.
$$
We shall use either of the above two descriptions of $\O$. 
Moreover, we 
shall often reverse the order 
so as to write $[v_1\wedge\ldots\wedge v_\ell]$ as $[v_\ell\wedge\ldots\wedge v_1]$. 

Now note that $\alpha$ can be divided into consecutive blocks as 
$$
\a = (\a_1, \dots , \a_{p_1}, \; \a_{p_1+1}, \dots , \a_{p_2}, \; \dots,
 \; \a_{p_{u-1}+1}, \dots , \a_{p_u}, \;
 \a_{p_{u}+1}, \dots , \a_{\ell} )
$$
so that $1\le p_1 < \cdots < p_u <  \ell$ and $\a_{p_i + 1}, \dots ,
\a_{p_{i+1}}$ are consecutive for $0 \le i \le u$, where 
$p_0 =0$ and $p_{u+1} = \ell$, by convention. If we further require that  $\a_{p_i + 1} - \a_{p_i} \ge 2$ for $i=1, \dots , u$, then the nonnegative integer $u$ and the ``jump spots'' $p_1, \dots , p_u$ are uniquely determined by $\a$. For example, if $\ell=7$ and $\alpha=(1,2,4,5,6,8,10)$, then $u=3$ and 
$(p_1, p_2, p_3) = (2,5,6)$. It is an easy consequence of the dimension formula (see, e.g., \cite[Lemma 2]{GT})
that if the dimension condition in \eqref{eq:6} holds at the ``jump spots'', then it holds everywhere else; in other words, 
 \begin{equation}
\label{eq:7}
\O=\{L\in G_{\ell}(A_{\ell}) : \dim (L\cap A_{p_i})\geq p_i \text{ for all } i=1, \dots , u \}.
\end{equation}
Hereafter, 
$u$ and $p_0, p_1, \dots , p_u, p_{u+1}$ will denote the unique integers satisfying 
 \begin{equation}
\label{eq:piu}
p_0:= 0 <  p_1 < \cdots < p_u <  p_{u+1} := \ell, \quad \text{and} \quad \a_{p_i + 1} - \a_{p_i} \ge 2 \text{ for } 1\le i\le u, 
\end{equation}
and moreover, $\a_{p_{i-1} + 1}, \dots ,
\a_{p_{i}}$ are consecutive for $1 \le i \le u+1$, that is, 
 \begin{equation}
\label{eq:consec}
\a_{p_i - j} = \a_{p_i}-j  \quad \text{for } \, 1\le i\le u+1 \, \text{ and } \, 1\le j < p_i - p_{i-1}.
\end{equation}
In particular, if $\a$ is \emph{completely consecutive}, i.e., if $u=0$, then \eqref{eq:7} shows that $\O$ coincides with the Grassmannian $G_{\ell}(A_{\ell}) $ of all $\ell$-dimensional subspaces of $A_{\ell}$. The other extreme is $u=\ell-1$, which means $\a_{i + 1} - \a_{i} \ge 2$ for $i=1, \dots , \ell-1$, and we will refer to such $\a$ as \emph{completely non-consecutive}. We now define a notion that will play an important role in the sequel. 

\begin{definition} 
\label{defSchub}
An element $f$ of $\bigwedge^{m-\ell}V$ is said to be \emph{Schubert decomposable} (with respect to the Schubert variety $\O$) if $f$ is decomposable, i.e., $f\ne 0$ and $f= f_1\wedge \cdots \wedge f_{m-\ell}$ for some $f_1, \dots , f_{m-\ell}\in V$, and moreover, 
 \begin{equation}
\label{eq:8}
\dim (V_f\cap A_{p_i}) =  \a_{p_i} - p_i  \quad \text{for all } i=1, \dots , u.
\end{equation}
\end{definition}
Note that if $\a$ is completely consecutive, then $u=0$ and condition \eqref{eq:8} is vacuously true. Thus,  in this case the notions of decomposable and Schubert decomposable elements are identical. However, in general, a decomposable element need not be Schubert decomposable. 
 
\subsection{Grassmann Codes and Schubert Codes}
\label{subsec:2.3}
Here, and hereafter, we will assume that the base field $\FF$ is the finite field $\Fq$. Then $G_{\ell,m}$ and $\O$ are finite and the number of ($\Fq$-rational) points in these varieties are $n$ and $n_{\a}$, which were given explicitly in \eqref{eq:1} and \eqref{eq:4}, respectively.  Fix an ordering $L_1, \dots , L_{n_{\a}}$ of the elements of $\O$ and representatives $P_1, \dots , P_{n_{\a}}$ in $\bigwedge^{\ell} V$ such that each $P_i$ is a decomposable element of the form $v_{\ell}\wedge \dots \wedge v_{1}$ with $v_i\in A_i$ for all $i=1, \dots , \ell$ and $L_j = V_{P_j}$ for $j=1, \dots , n_{\a}$. The choice we make here of ordering $v_i$ in the descending order in $i$ is merely a matter of convenience and will be found suitable when we use induction on $\ell$. Needless to say, the element $v_{\ell}\wedge \dots \wedge v_{1}$ differs only in sign from $v_1\wedge \dots \wedge v_{\ell}$. 
At any rate, we have a natural evaluation map
 \begin{equation}
\label{eq:9}
\bigwedge^{m-\ell} V \to \Fq^{n_{\a}} \quad \text{defined by} \quad  f\mapsto c_f, \quad 
\text{where} \quad c_f:= \left( f \wedge P_1, \dots , 
f\wedge P_{n_{\a}}\right).
\end{equation}
The \emph{Schubert code} $\C$ is defined as the image of this evaluation map. The Grassmann code $C(\ell,m)$ is a special case when $\alpha_i = m-\ell+i$ for  $i=1, \dots , \ell$. Note that a different choice of representatives results in a code that is monomially equivalent to $\C$. With this in view,  given any 
$f\in \bigwedge^{m-\ell} V $, instead of $f\wedge P_i$ 
we shall often write $f\wedge L_i$ or $f(L_i)$. 
This is an abuse of notation, but perfectly unambiguous when we are only interested in the vanishing or  nonvanishing of the scalar $f(L_i)$. This, for instance, is the case in the definition of the \emph{support} of $f$:
 \begin{equation}
\label{eq:10}
W(f) := \{ L\in \O : f(L)\ne 0\}.
 \end{equation}
It is clear that the cardinality of $W(f)$ is $\wt(c_f)$, i.e., the Hamming weight of the codeword $c_f$ of $\C$. In 
particular, 
 \begin{equation}
\label{eq:11}
d(\C) = \min\{|W(f)| : f\in \bigwedge^{m-\ell} V \text{ and $W(f)$ is nonempty}\}, 
 \end{equation}
where for any finite set $S$, we let $|S|$ denote the cardinality of $S$. 

We note that \eqref{eq:9} gives a surjective map of $\bigwedge^{m-\ell} V$ onto $\C$, but this map is, in general, not injective. In fact, its kernel is of dimension ${{m}\choose{\ell}} - k_{\a}$, where $k_{\a}$ is as in \eqref{eq:1}. Moreover, from an alternative expression for $k_{\a}$ given in \cite[eq. (4)]{GT}, it is readily seen that $k_{\a} < {{m}\choose{\ell}}$ if and only if $\O \ne G_{\ell, m}$. Note also that 
 \begin{equation}
\label{vanishing}
f \wedge L = 0 \; \text{ if $f\in \bigwedge^{m-\ell} V$ and $L \in \O$ are such that } V_f \cap L \ne \{0\}. 
  \end{equation}
This easily verifiable observation can be used tacitly in the sequel.

\section{Minimum distance of Schubert Codes}
\label{sec3}

For the remainder of this paper,  fix a prime power $q$, positive integers $\ell,m$ with $\ell < m$ and an $m$-dimensional vector space $V$ over $\Fq$. Also, let us fix 
a partial flag $A_1 \subset \dots \subset A_{\ell}$ of nonzero subspaces of $V$ and let 
$\a = (\a_1, \dots , \a_{\ell})$ be its dimension sequence. 
For any integer $j$, we set $A_j := \{0\}$ if $j\le 0$ and $A_j := V$ if $j>\ell$, by convention. Given any 
$v_1, \dots , v_r \in V$, we shall denote by $\langle v_1, \dots , v_r \rangle$ the linear subspace of $V$ generated by $v_1, \dots , v_r $. Likewise, if $L'$ is a subspace of $V$ and $v\in V$, then by $\langle L' , v\rangle$ we denote the subspace of $V$ generated by $v$ and the elements of~$L'$. 
Given a finite dimensional vector space $W$, a subspace of $W$ of dimension $\dim W -1$
may be referred to as a \emph{hyperplane} in $W$. 
We shall use the notation and terminology introduced in the previous section. In particular, given any 
$f\in\bigwedge^{m-\ell }V$, we denote by $c_f$ the corresponding codeword in the Schubert code $\C$. 

In case $\ell >1$, we will denote by 
$\alpha'$ the $(\ell-1)$-tuple $({\a}_1,\ldots,{\a}_{\ell-1})$, which is the dimension sequence of the truncated partial flag $A_1 \subset \dots \subset A_{\ell -1}$. 
Moreover, when $\ell >1$ and $f\in\bigwedge^{m-\ell }V$ is given, we put
$$
E:=\{x\in A_\ell: c_{f\wedge x} \text{ is the zero codeword in } \Cp \} \quad \text{and} \quad
F:= A_{\ell} \setminus E.
$$
It is clear that $E$ is a subspace of $A_{\ell}$. Naturally, $E$ and $F$ depend on $f$ and to make this dependence explicit, we could denote them by $E_f$ and $F_f$. However, in most situations there will be a fixed $f\in\bigwedge^{m-\ell }V$, and we will drop the subscript so as to simply write $E$ and $F$. 
%
%
	 The following lemma is 
	 a simple, but crucial, observation made by Xiang \cite{X}. We include a proof for the sake of completeness. 
	 
 \begin{lemma}\label{F}
	 	Assume that $\ell > 1$ and $f\in\bigwedge^{m-\ell }V$ is given. Let $E$ be as above and let $t$ be a nonnegative integer such that $\codim_{A_\ell}E \le t$. Then $A_{\ell-t}\subseteq E$.
 \end{lemma}	

\begin{proof}
If $t=0$ or $t \ge \ell$, then the result holds trivially. Thus, assume that $1\le t < \ell$. 
Suppose, on the contrary, there is some $x\in A_{\ell-t}\setminus E$. Then there are $x_i \in A_i$ for $1\le i\le \ell -1$ such that 
\begin{equation}
\label{eqstar}
f\wedge x \wedge x_{\ell-1} \wedge \dots \wedge x_1 \ne 0.
\end{equation} 
In particular, $x, x_{\ell-t} , \dots , x_{\ell-1}$ are linearly independent. 
Now if $y$ is any nonzero element of $\langle x, x_{\ell-t} , \dots , x_{\ell-1}\rangle$, then we can replace $x$ or some $x_j$ ($\ell - t\le j \le \ell-1$) by 
$y$ to obtain a basis of  $ \langle x, x_{\ell-t} , \dots , x_{\ell-1}\rangle$ consisting of $y$ and all except one among $ x, x_{\ell-t} , \dots , x_{\ell-1}$. So it follows from \eqref{eqstar} that 
$f\wedge y \wedge y_{\ell-1} \wedge \dots \wedge y_{1} \ne 0$ for some $y_i \in A_i$ ($1 \le i \le \ell-1$).
Consequently, $y\not\in E$. Thus 
$E \cap \langle x, x_{\ell-t} , \dots , x_{\ell-1}\rangle = \{0\}$. Hence, $\dim E \le \a_{\ell} - t - 1$, i.e., $\codim_{A_\ell} E \ge t +1$, which is a contradiction.  
\end{proof}

\begin{corollary}
\label{corF}
If $\codim_{A_\ell}E = 1$, then $L\not\subseteq A_{\ell-1}$ for every $L\in W(f)$.
\end{corollary}

\begin{proof}
Suppose, if possible, there is $L\in W(f)$ such that $L\subseteq A_{\ell-1}$. Then $L \subseteq E$, by Lemma \ref{F}. However, since $L\in W(f)$, there is $x\in L$ and $L'\in \Op$ such that $L = L'+\langle x \rangle$ and 
$(f\wedge x)(L')\ne 0$. But then $x\not \in E$, which is a contradiction. 
\end{proof}

The next lemma is also due to Xiang \cite{X}. 
We give a coordinate-free proof. 

\begin{lemma}
\label{x}
Assume that $\ell > 1$ and $f\in\bigwedge^{m-\ell }V$ is given. Let $\alpha'$ and $E$ 
be as defined above. Also, let 
$$
Z(\alpha',f)=\{(L',x)\in\Op \times A_\ell : (f\wedge x)(L')\neq 0\}
$$
and let $\phi:Z(\alpha',f)\longrightarrow W(f)$ be the map given by $(L',x)\mapsto \langle L',x\rangle$. Then
$\phi$ is well-defined and surjective. 
Moreover, given any $L \in W(f)$, the following holds. 
 \begin{enumerate}
\item[(i)]	If $L\not\subseteq A_{\ell-1}$ then $|\phi^{-1}(L)|=q^{\ell-1}(q-1)$.
	
\item[(ii)]	If $L\subseteq A_{\ell-1}$ and if $\, t:= \codim_{A_\ell}E $, then 
$|\phi^{-1}(L)|\leq q^{\ell-1} (q^t-1)$.
 \end{enumerate}	
 \end{lemma}	
%
%
%
 \begin{proof}
 It is clear that $\phi$ is well-defined (i.e., 
 $ \langle L',x\rangle  \in W(f)$ 
 whenever $(L',x)$ is in $Z(\alpha', f)$) 
 and that $\phi$ is surjective. Now let $L \in W(f)$. 
 
 (i) Suppose 
 $L\not\subseteq A_{\ell-1}$. Since $L \in \O$, we see that $\dim (L\cap A_{\ell-1}) = \ell -1$. Now if 
 $(L',x) \in \phi^{-1}(L)$, then $L'$ is an $(\ell-1)$-dimensional subspace of $L \cap A_{\ell -1}$ and hence 
 $L'=L\cap A_{\ell-1}$. On the other hand, $x$ can be an arbitrary element of $L \setminus L'$. Thus $\left| \phi^{-1}(L) \right| = q^{\ell} - q^{\ell-1} = q^{\ell-1}(q-1)$. 
 
 (ii) Suppose 
 $L \subseteq A_{\ell-1}$ and 
 $t:= \codim_{A_\ell}E $. 
   Let $(L',x) \in \phi^{-1}(L)$. 
   Observe~that 
\begin{equation}
\label{claim}
 L'\cap A_{\ell-t}=L\cap A_{\ell-t} .
 \end{equation}
 Indeed, the inclusion $\subseteq$ is obvious, whereas if there exists $u\in( L\cap A_{\ell-t}) \setminus L'$, then 
 $L = \langle L', u\rangle$ and by Lemma \ref{F}, $u \in E$, which implies that $(f\wedge u) (L') = 0$ and hence $f(L) = 0$, which contradicts the assumption that $L\in W(f)$.

 Now let $L_t:= L\cap A_{\ell-t} $.
  From \eqref{claim}, we see that $L'$ is necessarily a hyperplane in $L$ containing $L_t$. 
 The number of such hyperplanes is equal to the number of hyperplanes in $L/L_t$. Thus if $r:= \dim L_t$ and $N':= (q^{\ell - r} - 1)/(q-1)$, then we see that there are at most $N'$ choices for $L'$. 
 Note that since $L\in \O$, we have $r \ge \ell -t$ and hence $N' \le (q^t-1)/(q-1)$. Moreover, 
 $x$ has to be in 
 $L\setminus L'$ and so there are at most $q^{\ell} - q^{\ell - 1}$ possibilities for $x$. It follows that $|\phi^{-1}(L)|\leq q^{\ell-1} (q^t-1)$.
  \end{proof} 
  
\begin{lemma}
\label{prexx}
Assume that $\ell > 1$. Let $f\in\bigwedge^{m-\ell }V$ be such that $c_f\ne 0$ and let $E$ and $F$ 
be the corresponding sets as defined above. Also, let $t:= \codim_{A_{\ell}}E$. Then 
\begin{equation}
\label{w1w2}
 \wt(c_f)\geq \frac{1}{ q^{\ell-1}(q-1)} \sum\limits_{x\in F \setminus A_{\ell-1} }  \wt(c_{f\wedge x}) +\frac{1}{q^{\ell-1}(q^t-1)} \sum\limits_{x\in F\cap A_{\ell-1}} \wt(c_{f\wedge x}).
\end{equation}
Moreover, the inequality above is strict if the inequality in part (ii) of Lemma \ref{x} is strict for some $L\in W(f)$ with $L\subseteq A_{\ell-1}$. 
 \end{lemma}	

\begin{proof}
Let $\alpha',  \; Z(\alpha',f)$, and 
$\phi:Z(\alpha',f)\longrightarrow W(f)$ be as in Lemma \ref{x}. Then 
$$
 \left| Z(\alpha',f) \right| = 
 \sum\limits_{L\in W(f)} |\phi^{-1}(L)|=
 \mathop{\sum_{L\in W(f)}}_{L\nsubseteq A_{\ell-1}} |\phi^{-1}(L)|+
 \mathop{\sum_{L\in W(f)}}_{L\subseteq A_{\ell-1}}|\phi^{-1}(L)|.
$$
On the other hand, considering the fibres of the 
projection $Z(\alpha',f)\to F$, we obtain
$$
 \left| Z(\alpha',f) \right| = \sum_{x\in F} \wt (c_{f\wedge x}) =  \sum_{x\in F\setminus A_{\ell -1}} \wt (c_{f\wedge x}) + \sum_{x\in F \cap A_{\ell -1} } \wt (c_{f\wedge x})  .
 $$
 Now observe that for any $L' \in \Op$ and $x\in A_{\ell}$, we have 
 $$
 L' + \langle x \rangle \in W(f) \text{ and }  L' + \langle x \rangle \nsubseteq A_{\ell -1} \Longleftrightarrow 
 x\in F \setminus A_{\ell -1} \text{ and } L' \in W(f\wedge x).
 $$
 This implies the first equality below, which in turn, yields the second inequality. 
 $$
  \mathop{\sum_{L\in W(f)}}_{L\nsubseteq A_{\ell-1}} |\phi^{-1}(L)| = \sum_{x\in F\setminus A_{\ell -1}} \wt (c_{f\wedge x}) \ \text{ and} \  
\mathop{\sum_{L\in W(f)}}_{L\subseteq A_{\ell-1}}|\phi^{-1}(L)| = 
\sum_{x\in F \cap A_{\ell -1} } \wt (c_{f\wedge x}).
$$
Hence, if we let $\theta_1:= |\{L\in W(f): L\nsubseteq A_{\ell-1}\}|$  and $\theta_2:= |\{L\in W(f): L\subseteq  A_{\ell-1}\}|$, then $\wt(c_f) = \theta_1 +\theta_2$ and  from Lemma \ref{x}, we see that 
$$
\sum_{x\in F\setminus A_{\ell -1}} \wt (c_{f\wedge x}) =  \theta_1 q^{\ell-1}(q-1) \ \text{ and} \  
\sum_{x\in F \cap A_{\ell -1} } \wt (c_{f\wedge x}) \leq \theta_2 q^{\ell-1}(q^t-1). 
$$
This implies \eqref{w1w2}. 
Moreover,  if 
$|\phi^{-1}(L)| < q^{\ell-1}(q^t-1)$ for some $L\in W(f)$ with $L\subseteq A_{\ell-1}$, then it is clear that the inequality in \eqref{w1w2} is strict. 
\end{proof}

For ease of reference, we state 
the following result for which a short proof is given in \cite[Prop. 5.2]{GL}, while an 
alternative proof is given in \cite[Thm. 1]{X}. Yet another proof will be sketched in Remark \ref{OtherPfeasyineq}. 

\begin{proposition}
\label{easyineq}
$d(\C) \le \Q$. 
\end{proposition}

We are now ready to show that the MDC holds in the affirmative, in general. We shall also see that the proof also gives us some information about the minimum weight codewords of $\C$. 

\begin{theorem}
\label{xx}
	$d(\C) = \Q$. Moreover, if $\ell > 1$ and $f\in \bigwedge^{m-\ell}V$ is such that $c_f$ is a minimum weight codeword in $\C$, then $c_{f\wedge x}$  is a minimum weight codeword in $\Cp$ for every $x\in F$ and furthermore, 
we must have either (i) $t =1$ and $t'=0$, or (ii) $t' = t \ge 2$,  $\a_{\ell} - \a_{\ell -1} = 1$, and 
equality holds in \eqref{w1w2}.  Here $\a' , \, E$ and $F$ are as 
before, while $t:= \codim_{A_{\ell}}E$ 
and $t':= \codim_{A_{\ell-1}}(E \cap A_{\ell-1})$.
%
\end{theorem}

\begin{proof}
In view of Proposition \ref{easyineq}, in order to show that $d(\C) = \Q$, it suffices to show that 
\begin{equation}
\label{otherineq}
\wt(c_f) \ge \Q \quad \text{for every $f\in \bigwedge^{m-\ell}V$ such that } c_f \ne 0.
\end{equation}
We now proceed to prove \eqref{otherineq} by induction on $\ell$ ($1 \le \ell < m$). The initial case can be deduced from facts about Grassmann codes or first order projective Reed-Muller codes, but we will give a direct and self-contained proof. The induction step will make use of the above lemma. 

First, suppose $\ell =1$ and $f\in \bigwedge^{m-\ell}V$. Then $f$ is necessarily decomposable, say  $f=f_1\wedge\dots\wedge f_{m-1}$ for some $f_1, \dots , f_{m-1} \in V$. 
Further, suppose  $c_f\ne 0$.  Then we must have  $A_1\nsubseteq V_f = \langle f_1,\ldots, f_{m-1}\rangle$. Consequently, 
$A_1 + V_f = V$ and hence $\dim A_1 \cap V_f = \a_1 -1$. Noting 
that 
$W(f) = \{ \langle x \rangle : x \in A_1 \setminus (A_1\cap V_f) \}$, we obtain 
$$
\wt(c_f) = |W(f)| = \frac{ q^{\a_1} - q^{\a_1 - 1}}{q-1} = q^{\a_1 - 1} = \Q.
$$

Next, suppose $1 < \ell < m$ and \eqref{otherineq} holds for positive 
values of $\ell$ smaller than the given one. 
Note that  $t \ge 1$ since $c_f \ne 0$. 
Note also that $t':= \codim_{A_{\ell-1}}(E \cap A_{\ell-1})$ satisfies $t'\le t$ 
since the inclusion $A_{\ell-1}\hookrightarrow A_{\ell}$ 
induces an injective homomorphism 
$A_{\ell-1}/E\cap A_{\ell-1}\hookrightarrow A_{\ell}/E$.
We shall now divide the proof into two cases. 

\medskip
\noindent
\textbf{Case 1.}  $t=1$. 

In this case, by Lemma \ref{F},  $A_{\ell-1}\subseteq E$. Hence $F \setminus A_{\ell-1}=F$ 
and 
$F\cap A_{\ell-1} = \emptyset$. 
It follows that 
$|F \setminus A_{\ell-1}| = |F| = |A_{\ell}| - |E| = q^{{\a}_\ell} - q^{{\a}_\ell-1}$. 
Moreover, by the induction hypothesis, 
$\wt(c_{f\wedge x}) \geq q^{\delta(\alpha')} $, 
for every $x\in F$. Now \eqref{w1w2} reduces to 
$$
\wt(c_f) \ge \frac{1}{ q^{\ell-1}(q-1)}\sum\limits_{x\in F}\wt(c_{f\wedge x})
	       \ge \frac{1}{ q^{\ell-1}(q-1)} (q^{{\a}_\ell} - q^{{\a}_\ell-1})  
	       q^{\delta(\alpha')}\\
	       = \Q.
$$
Thus, \eqref{otherineq} is proved in this case. Also, it is clear that if 
$\wt (c_{f\wedge x}) > q^{\delta(\alpha')} $ for some $x\in F$, then $\wt(c_f) > q^{\delta(\alpha)}$. 
Note also that $t'=0$ in this case, since $A_{\ell-1}\subseteq E$. 

\medskip
\noindent
\textbf{Case 2.}  $t\geq 2$. 

First note that, with $t$ and $t'$ as above, 
$$
|F\setminus A_{\ell-1}| = \left( |A_{\ell}| - |E|\right) - \left(|A_{\ell-1}| - |E \cap A_{\ell-1}| \right) = \left(q^{{\a}_\ell}-q^{{\a}_\ell- t} \right)- \big( q^{{\a}_{\ell-1}}-q^{{\a}_{\ell-1}-t'}\big).
$$
We will now consider three different subcases as follows.

\textbf{Subcase $\mathbf{2.1}$.}  
${\a}_\ell-{\a}_{\ell-1} \geq 2$. 
 
Here, using the expression for $|F\setminus A_{\ell-1}|$ obtained above, we see that 
$$
|F\setminus A_{\ell-1}|  - \left(q^{{\a}_{\ell}} - q^{ {\a}_{\ell}-1} \right) = \left(q^{{\a}_{\ell}-1} - q^{{\a}_{\ell} - t} - q^{ {\a}_{\ell-1}} \right) + q^{ {\a}_{\ell} -t' } > 0,
$$
where the last inequality follows by noting that $q^{{\a}_{\ell}-1} \ge 2 q^{{\a}_{\ell}-2} \ge q^{{\a}_{\ell}-t} + q^{{\a}_{\ell-1} }$, since $t\ge 2$ and $ {\a}_\ell-{\a}_{\ell-1} \geq 2$.  Hence, by \eqref{w1w2} and the induction hypothesis, 
$$
\wt(c_f) \geq \frac{1}{ q^{\ell-1}(q-1)} |F\setminus A_{\ell-1}|\, q^{\delta(\alpha')} > 
\frac{1}{ q^{\ell-1}(q-1)} (q^{{\a}_\ell} - q^{{\a}_\ell-1})q^{\delta(\alpha')}
= \Q.
$$ 
Thus, we obtain \eqref{otherineq} with, in fact,  a strict inequality.

 \smallskip 

\textbf{Subcase $\mathbf{2.2}$.}  
${\a}_\ell-{\a}_{\ell-1}=1$ and $t'\ne t$.
 
Here, $t' \le t-1$ and so using the expression for $|F\setminus A_{\ell-1}|$ obtained earlier, we see that 
$
|F\setminus A_{\ell-1}|  - \left(q^{{\a}_{\ell}} - q^{ {\a}_{\ell}-1} \right) = \big(q^{{\a}_{\ell}-1-t'} - q^{{\a}_{\ell} - t} \big)  \ge  0$, 
and also that strict inequality holds when $t' = 0$. 
Hence, by \eqref{w1w2} and the induction hypothesis,
$$
\wt(c_f) >  
\frac{1}{ q^{\ell-1}(q-1)} (q^{{\a}_\ell} - q^{{\a}_\ell-1})q^{\delta(\alpha')}
= \Q,
$$ 
where the above inequality is strict either because $t'>0$ in which case the second summation in  \eqref{w1w2} is nonempty and contributes a positive term or because $t'=0$ in which case 
$|F\setminus A_{\ell-1}|  > \left(q^{{\a}_{\ell}} - q^{ {\a}_{\ell}-1} \right)$. 
This yields \eqref{otherineq}, with  a strict inequality. 

 \smallskip 

\textbf{Subcase $\mathbf{2.3}$.}  
${\a}_\ell-{\a}_{\ell-1}=1$ and $t' = t$.
 
In this subcase of Case 2, we readily see that 
$$
|F\setminus A_{\ell-1}|  =  q^{{\a}_{\ell} - t-1} (q^t - 1) (q-1) \quad \text{and} \quad 
|F \cap  A_{\ell-1}|  = q^{{\a}_{\ell} - t - 1 } (q^t - 1). 
$$
Hence, from \eqref{w1w2} and the induction hypothesis, we obtain
$$
\wt(c_f) \ge  
\frac{q^{{\a}_{\ell} - t-1} (q^t - 1) }{ q^{\ell-1}} q^{\delta(\alpha')}  + 
\frac{q^{{\a}_{\ell} - t-1} }{ q^{\ell-1}} q^{\delta(\alpha')} 
= \Q.
$$ 
Also, it is clear that if 
$\wt (c_{f\wedge x}) > q^{\delta(\alpha')} $ for some $x\in F$ or if the inequality in \eqref{w1w2} is strict, then $\wt(c_f) > q^{\delta(\alpha)}$. 

Thus, \eqref{otherineq} is proved in all cases. Hence, by induction on $\ell$ we conclude that $d(\C) = \Q$. The remaining assertions in the statement of the theorem are also clear from the proof. 
\end{proof}

\section{Annihilators of Decomposable Elements} 
\label{newsec4} 

Now that we know the minimum distance of Schubert codes, it is natural to ask for a classification as well as enumeration of the minimum weight codewords. In this section, we shall take some preliminary steps towards such a classification by analysing the intersections of annihilators of decomposable elements of 
$\bigwedge^{m-\ell}V$ with the constituent subspaces of the partial flag defining the given Schubert variety.  

As in 
Section \ref{sec3}, 
integers $\ell, m$ with $1\le \ell < m$, an $m$-dimensional vector space $V$ over $\Fq$, a partial flag 
$A_1 \subset \dots \subset A_{\ell}$ of nonzero subspaces of $V$ with dimension sequence 
$\a = (\a_1, \dots , \a_{\ell})$ will be kept fixed throughout this section. 
Also, recall (from \S 2.2) 
that for any  $f \in\bigwedge^{m-\ell}V$, by $V_f$ we denote the annihilator of $f$. 


%

 \begin{lemma}\label{S}
 	Suppose $f \in\bigwedge^{m-\ell}V$ is a decomposable element with $c_f\ne 0$. Then  
$$
{\a}_i-\ell \, \leq \, \dim(V_f\cap A_i) \, \leq \, {\a}_i-i  \quad \text{ for all } i=1, \dots , \ell.
$$
In particular, $\dim(V_f\cap A_\ell)= {\a}_\ell-\ell$. 
 \end{lemma}
 
 \begin{proof}
 Since $f$ is decomposable, $\dim V_f = m- \ell$.  Hence, for $1\le i \le \ell$, 
 $$
 	m \geq \dim(V_f+A_i) 	=\dim(V_f)+\dim(A_i)-\dim(V_f\cap A_i)
 	= m-\ell+{\a}_i-\dim(V_f\cap A_i),
 $$
and thus $\dim(V_f\cap A_i)\geq {\a}_i-\ell$. 
Further, since $c_f\ne 0$, there are $x_j \in A_j$ for $j=1, \dots , \ell$ such that $f\wedge x_\ell\wedge\ldots\wedge x_1\neq 0$. Consequently, for each $i =1, \dots , \ell$, there are linearly independent $x_1, \dots , x_i \in A_i$ such that $f\wedge x_i\wedge\ldots\wedge x_1\neq 0$ and therefore
$V_f \cap \langle x_1, \dots , x_i\rangle = \{0\}$. This implies that  $\dim(V_f\cap A_i)\leq {\a}_i-i$.
 \end{proof}

It turns out that the attainment of the upper bound given in Lemma \ref{S} for $\dim(V_f\cap A_i)$ has a nice characterization when $i=\ell -1$.

\begin{lemma}
\label{P}	
Assume that $\ell > 1$. Let $f \in\bigwedge^{m-\ell}V$ be a decomposable element with 
$c_f\ne 0$, and 
$E$ 
the corresponding subspace as in \S 
\ref{sec3}. Then: 
$$
\dim(V_f\cap A_{\ell-1}) = {\a}_{\ell-1}-(\ell-1) \Longleftrightarrow \codim_{A_\ell}E =1.
$$
\end{lemma}

\begin{proof}
Since $c_f \ne 0$, there are  $x_i\in A_i$ for $1\le i \le \ell$ such that $f\wedge (x_\ell\wedge\ldots\wedge x_1)\neq 0$. Consequently, $V_f\cap \langle x_1,\ldots, x_{\ell}\rangle=\{ 0 \}$ and $x_{\ell} \not\in E$. In particular, $E \ne A_{\ell}$. 

Suppose $\dim(V_f\cap A_{\ell-1}) = {\a}_{\ell-1}-(\ell-1)$. Observe that $x_1,\ldots, x_{\ell-1}\in E$. 
Indeed, if $x_i\notin E$ for some $i\leq \ell-1$, then there exist 
 $y_j\in A_j$ for $j=1, \dots , \ell -1$ such that $(f\wedge x_i)\wedge(y_{\ell-1}\wedge\ldots\wedge y_1)\neq 0$;  consequently, $\langle x_i,y_1,\ldots, y_{\ell-1}\rangle$ is a subspace of $A_{\ell-1}$ of dimension ${\ell}$ 
 such that $V_f\cap \langle x_i,y_1,\ldots, y_{\ell-1}\rangle=\{0\}$, and this yields a contradiction. 
It follows  
that $ (V_f \cap A_{\ell}) +  \langle x_1,\ldots, x_{\ell-1}\rangle$ is an $(\a_{\ell} -1)$-dimensional subspace of  $E$. Since $E \ne A_{\ell}$, we conclude that $\codim_{A_\ell}E=1$. 

Conversely, 
suppose $\codim_{A_\ell}E =1$.  
Note that $\dim(V_f\cap A_{\ell-1}) \le {\a}_{\ell-1}-(\ell-1)$, thanks to Lemma \ref{S}. 
In case $\dim(V_f\cap A_{\ell-1})<{\a}_{\ell-1}-(\ell-1)$, there exists $y_\ell\in A_{\ell-1}$ such that  $y_\ell \not\in (V_f\cap A_{\ell-1}) + \langle x_1,\ldots,\ x_{\ell-1}\rangle$. Clearly,  $f\wedge (y_\ell\wedge x_{\ell-1}\wedge\ldots\wedge x_1)\neq 0$. It follows that $x_{\ell-1}$ and $y_\ell$ are linearly independent elements of $A_{\ell-1}$ and neither of them is in  $E$.  Changing $x_{\ell-1}\wedge y_\ell$ to $x_{\ell-1}\wedge z$ or $z\wedge y_\ell$ for any nonzero $z\in \langle x_{\ell-1}, y_\ell\rangle$, we see that $\langle x_{\ell-1},y_\ell\rangle\cap E=\{0\}$, and so 
$\codim_{A_\ell}E>1$, which is a contradiction.
	\end{proof}
%
%
%
	
\section{Schubert Decomposability and Minimum Weight Codewords}
\label{newsec5}

We will continue to use the notation and terminology of the previous two sections. 
For the dimension sequence $\a = (\a_1, \dots , \a_{\ell})$ of the fixed partial flag 
$A_1 \subset \dots \subset A_{\ell}$,  
we let $u$ and $p_0, p_1, \dots , p_u, p_{u+1}$ denote the unique integers satisfying \eqref{eq:piu} and \eqref{eq:consec}. Recall that $f\in \bigwedge^{m-\ell}V$ is said to be Schubert decomposable (w.r.t. $\O$) if 
$\dim (V_f \cap A_{p_i}) = \a_{p_i} - p_i$ for $i=1, \dots , u$. Note that this equality for dimension also holds for $i=u+1$, thanks to Lemma \ref{S}. 
%
%

We shall now proceed to relate 
Schubert decomposability with the minimum weight codewords of $\C$. 
We begin with a simple and basic observation. 

\begin{lemma}
\label{lem:SchubNonzero}
Let $f\in \bigwedge^{m-\ell} V $ be Schubert decomposable. Then $c_f$ is nonzero. 
\end{lemma}

\begin{proof}
Since $\dim  (V_f\cap A_{p_1}) =  \a_{p_1} -p_1$, by extending a basis of $ V_f\cap A_{p_1}$ to  $A_{p_1}$, we can  find a $p_1$-dimensional 
subspace $L_1$ of $A_{p_1}$ such that $L_1 \cap V_f = \{0\}$. 
Now since 
$\a_1, \dots , \a_{p_1}$ are consecutive, we see that 
$$
\dim (L_1 \cap A_{p_1-j}) \ge \dim (L_1 \cap A_{p_1-j+1}) - 1\quad \text{for each } 1 \le j < p_1
$$
and this 
implies that $\dim (L_1 \cap A_{i}) \ge i$ for $i=1, \dots , p_1$. Hence, we can choose $x_i \in A_i$ for $i=1, \dots , p_1$ such that $\{x_1, \dots , x_{p_1}\}$  forms a basis of $L_1$. 
Next, observe that $\dim  (V_f\cap A_{p_2})+ L_1 =  \a_{p_2} - p_2 + p_1$.  Now since 
$\a_{p_1+1}, \dots , \a_{p_2}$ are consecutive, by arguing as before, we can find  $x_i \in A_i$ for $i=p_1+1, \dots , p_2$ such that $\{x_1, \dots , x_{p_2}\}$ forms a basis of a  
$p_2$-dimensional subspace $L_2$ of $A_{p_2}$ such that $L_2 \cap V_f = \{0\}$ and $L_1 \subset L_2$. 
Continuing in this manner, we obtain linearly independent $x_1, \dots , x_{\ell} \in V$ such that $x_i \in A_i$ for $i=1, \dots , \ell$ and $V_f \cap \langle x_1, \dots , x_{\ell} \rangle = \{0\}$. Consequently, $L:=  [x_1 \wedge \dots \wedge x_{\ell}]$ is in $\O$ and since $f$ is decomposable, we must have $f(L)\ne 0$. Thus,  $c_f \ne 0$.   
\end{proof}

Our next result is a refined version of Lemma \ref{x} with an additional hypothesis of Schubert decomposability. 

\begin{lemma}
\label{x2}
Assume that $\ell > 1$ and $f\in\bigwedge^{m-\ell }V$ is Schubert decomposable. Let $\alpha', \, E,  \, 
Z(\alpha',f)$ 
and $\phi:Z(\alpha',f)\longrightarrow W(f)$ be as in Lemma \ref{x}, and 
$t:= \codim_{A_{\ell}} E$.  Then  $t=1$ or $t = \ell - p_u$. Moreover, 
\begin{equation}
\label{RefinedFibre}
|\phi^{-1}(L)| =  q^{\ell-1} (q^t-1) \quad \text{for every } \, L \in W(f).
\end{equation} 	
\end{lemma}
  
\begin{proof}
If 
$t=1$, then \eqref{RefinedFibre} 
follows from Corollary \ref{corF} and part (i) of Lemma \ref{x}. Now suppose $t>1$. Then 
from Lemma  \ref{S} and \ref{P},  we obtain 
$\dim (V_f \cap A_{\ell -1}) = \a_{\ell -1} - \ell$. 
Now $p_u \le \ell -1$, and if we had $p_u = \ell -1$, then 
$
\dim (V_f \cap A_{p_u}) = \a_{\ell -1} - \ell < \a_{p_u} - p_u, 
$
which contradicts that $f$ is Schubert decomposable. So we must have $p_u < \ell - 1$. 
Moreover, since $\a_{p_u+1}, \dots , \a_{\ell -1}$ are consecutive, if for some $j$ with $p_u+1 \le j \le \ell-1$, we had $\dim (V_f \cap A_j) \ge \a_j - j$, then 
 we would obtain 
$$
\dim (V_f \cap A_{\ell -1}) \ge  \a_j - j = \a_{\ell -1} - (\ell - 1 - j) - j = \a_{\ell -1} - (\ell -1),
$$
which is a contradiction. Thus,  for each $j=p_u+1, \dots , \ell-1$, we have 
\begin{equation}
\label{uplusone}
\dim (V_f \cap A_j) < \a_j - j,  \text{ and in particular, } 
\dim (V_f \cap A_{p_u+1}) \le \a_{p_u+1} - p_u - 2.
\end{equation}
Now since $\dim (V_f \cap A_{p_i}) = \a_{p_i} - p_i$ for $i=1, \dots , u$, by arguing as in the proof of Lemma \ref{lem:SchubNonzero}, we obtain linearly independent elements $x_1, \dots , x_{p_u} \in V$ such that 
$V_f \cap \langle x_1, \dots , x_{p_u} \rangle = \{0\}$ and $\langle x_1, \dots , x_{p_i} \rangle \subseteq A_{p_i}$ for each $i=1, \dots , u$. 
Hence the sum $(V_f \cap A_{p_u+1}) +  \langle x_1, \dots , x_{p_u} \rangle $ is a subspace of $A_{p_u+1}$ (as well as $A_{\ell-1}$) of dimension $\le \a_{p_u+1} - 2$, and so we can find $y_1, y_2$ in $A_{p_u+1}$ and 
$y_3, \dots , y_{\ell - p_u}$ in $A_{\ell -1}$ such that $x_1, \dots , x_{p_u}, y_1, y_2, \dots, y_{\ell - p_u}$ are linearly independent and moreover, 
$
V_f \cap  \langle x_1, \dots , x_{p_u}, y_1, y_2, \dots, y_{\ell - p_u} \rangle  = \{0\}. 
$
Since $f$ is decomposable, it follows that $f\wedge x_1 \wedge \dots \wedge x_{p_u} \wedge y_1 \wedge \dots \wedge y_{\ell - p_u} \ne 0$. Consequently, $y_1\in A_{p_u+1} \setminus E$ and so $A_{p_u+1} \nsubseteq E$. 
Hence from Lemma \ref{F}, we see that $\codim_{A_{\ell}} E > \ell - (p_u+1)$, i.e., $\codim_{A_{\ell}} E \ge \ell - p_u$. We will now proceed to show that $\codim_{A_{\ell}} E = \ell - p_u$. 
To this end, observe that $A_{p_u} \subseteq E$. Indeed, if there exists $x\in A_{p_u} \setminus E$, then $f\wedge x\wedge z_1 \wedge \dots \wedge z_{\ell -1} \ne 0$ for some $z_i\in A_i$ ($1\le i \le \ell -1$). In particular, 
 $f\wedge x\wedge z_1 \wedge \dots \wedge z_{p_u} \ne 0$, and hence $V_f \cap \langle  x, z_1, \dots , z_{p_u} \rangle =\{ 0\}$. This implies that $\dim (V_f \cap  A_{p_u}) \le \a_{p_u} - p_u - 1$, which contradicts the assumption that $f$ is Schubert decomposable. Thus, $A_{p_u} \subseteq E$ and hence 
$(V_f \cap A_{\ell}) + A_{p_u}\subseteq E$. Consequently, 
$$
(\a_{\ell} - \ell) + \a_{p_u} - (\a_{p_u} - {p_u} ) \le \dim E, \quad \text{that is,} \quad   \codim_{A_{\ell}} E \le \ell - p_u. 
$$
Thus, we have proved that $t = \ell - p_u$. 

Now fix any $L\in W(f)$. Then $L\in \O$ with $f(L)\ne 0$. 
Note that $L\subseteq A_{\ell}$. Let $L_t: = L\cap A_{\ell - t} = L \cap A_{p_u}$. Observe that $\dim L_t \ge p_u$, since $L\in \O$. 
Also, since $L\cap V_f =\{0\}$, we see that $(V_f \cap A_{p_u}) \cap L_t =\{0\}$. Hence, 
the 
Schubert decomposability of $f$ implies that 
$\a_{p_u} - p_u + \dim L_t \le \a_{p_u}$, i.e., $\dim L_t \le p_u$. Thus $\dim L_t = p_u$. Next, we claim that  for any 
$L'\in G_{\ell-1}(V)$ and $x\in A_{\ell}$, 
$$
(L',x) \in  \phi^{-1}(L) \Longleftrightarrow  \;  L'   \text{ is a hyperplane in $L$ containing $L_t$ and } x\in L\setminus L'.
$$
The implication $\Rightarrow$ is clear because 
we have seen in the proof of Lemma \ref{x} that if 
$(L',x) \in  \phi^{-1}(L)$, then $L'\cap A_{\ell - t} = L_t$. For the other implication, suppose $L'$ is a hyperplane in $L$ and $x\in L\setminus L'$. Then it is clear that $L = \langle L', x\rangle$. 
Further, suppose $L_t \subseteq L'$. Now since $\ell - t = p_u \ge p_i$ for $i=1, \dots , u$, we find 
$$
L \cap A_{p_u} =L_t \subseteq L' \Longrightarrow L \cap A_{p_u} =L'\cap A_{p_u}
\Longrightarrow L \cap A_{p_i} =L'\cap A_{p_i} \text{ for $i=1, \dots , u$}.
$$
Hence in view of \eqref{eq:7}, we see that $L' \in \Op$ and 
thus the claim is proved. As a consequence, we see that $|\phi^{-1}(L)| = N' (q^\ell - q^{\ell -1})$, where $N'$ is exactly the number of hyperplanes in $L$ containing $L_t$. Now $N' = (q^{\ell - (\ell -t)} - 1)/(q-1)$, exactly as in the proof of Lemma \ref{x}. This yields 
the desired formula for $|\phi^{-1}(L)|$.
\end{proof}

\begin{theorem}
\label{thmS1}
If $f\in\bigwedge^{m-\ell }V$ is Schubert decomposable, then $c_f$ is a minimum weight codeword of $\C$. 
\end{theorem}

\begin{proof}
We use induction on $\ell$. If $\ell =1$, then the desired result follows from Lemma~\ref{lem:SchubNonzero} since 
$\Omega_{\a}(1, m) = \PP(A_1)$ and $C_{\alpha}(1,m)$ is the $q$-ary simplex code of length $\left(q^{\a_1} -1 \right)/(q-1)$ and dimension $\a_1$, and hence every nonzero codeword of $C_{\alpha}(1, m)$ is a minimum weight codeword. 

Now suppose $\ell > 1$ and that the result holds for values of $\ell$ smaller than the given one. Let $f\in\bigwedge^{m-\ell }V$ be Schubert decomposable, and let $E$ and $F$ be the corresponding subsets of $A_{\ell}$ as in Section \ref{sec3}. Given any $x\in F$, we note that $g:= f\wedge x$ is a decomposable element of $\bigwedge^{m-(\ell-1) }V$ satisfying $V_f \subset V_g$ and $c_g\ne 0$. 
In particular,  we find $V_{f} \cap A_{p_{i}} \subseteq V_{g}\cap A_{p_{i}}$ for $1\le i \le u$. On 
the other hand, $ \dim V_{g}\cap A_{p_{i}} \le {\alpha }_{p_{i}} - p_{i} = \dim V_{f} \cap A_{p_{i}}$ 
for $1\le i \le u$, thanks to 
Lemma \ref{S} and the Schubert decomposability of $f$. It follows that $g$ is Schubert decomposable and hence by the induction hypothesis, $c_g$ is a minimum weight codeword of $\Cp$. 
Now, proceeding as in the proof of Lemma \ref{prexx}, except for applying Lemma \ref{x2} in place of part (ii) of Lemma \ref{x}, we see that 
$$
\sum_{x\in F\setminus A_{\ell -1}}  \wt (c_{f\wedge x}) =  \theta_1 q^{\ell-1}(q-1) \ \text{ and} \  
\sum_{x\in F \cap A_{\ell -1} } \wt (c_{f\wedge x}) = \theta_2 q^{\ell-1}(q^t-1), 
$$
where  $\theta_1:= |\{L\in W(f): L\nsubseteq A_{\ell-1}\}|$  and $\theta_2:= |\{L\in W(f): L\subseteq  A_{\ell-1}\}|$.  
Since $\wt(c_f) = \theta_1 +\theta_2$ and since 
$f\wedge x$ is of weight ${\Qp}$ for every $x\in F$, we obtain 
\begin{equation}
\label{RefinedWt}
 \wt(c_f) =  \frac{1}{ q^{\ell-1}(q-1)} \left| F \setminus A_{\ell-1} \right| \Qp +\frac{1}{q^{\ell-1}(q^t-1)} \left| F\cap A_{\ell-1} \right| \Qp.
\end{equation} 
In case $t=1$, this gives
$$
 \wt(c_f) =  \frac{1}{ q^{\ell-1}(q-1)} \left| F  \right| \Qp = \frac{1}{ q^{\ell-1}(q-1)} (q^{{\a}_\ell} - q^{{\a}_\ell-1})q^{\delta(\alpha')}
= \Q.
$$
Now suppose $t > 1$. Then $t = \ell - p_u$ by Lemma \ref{x2} and so $p_u < \ell$. Consequently, ${\a}_{\ell} - {\a}_{\ell-1} = 1$. Thus, $A_{\ell-1}$ is a hyperplane in $A_{\ell}$, and therefore
$$
\dim E \cap A_{\ell-1} \ge \dim E - 1. 
$$
Moreover, in view of Lemmas \ref{P} and \ref{S}, we see that 
$$
\dim(V_f\cap A_{\ell-1}) = {\a}_{\ell-1}-\ell = {\a}_{\ell}-\ell - 1 = 
\left( \dim V_f\cap A_{\ell} \right) - 1. 
$$
Hence we can find some $z \in V_f\cap A_{\ell}$ such that $z\not\in V_f\cap A_{\ell-1}$. Since $V_f \subseteq E$, we see that $z\in E \setminus (E \cap A_{\ell -1})$ and therefore 
$$
\dim E \cap A_{\ell-1} \le \dim E - 1. 
$$
It follows that $\dim E \cap A_{\ell-1} = \dim E - 1$, and hence $t':= \codim_{A_{\ell-1}}(E \cap A_{\ell-1})$ is equal to $t$. 
Consequently, as in Subcase 2.3 of the proof of Theorem \ref{xx}, 
$$
|F\setminus A_{\ell-1}|  =  q^{{\a}_{\ell} - t-1} (q^t - 1) (q-1) \quad \text{and} \quad 
|F \cap  A_{\ell-1}|  = q^{{\a}_{\ell} - t - 1 } (q^t - 1). 
$$
Using this together with \eqref{RefinedWt}, we obtain $\wt(c_f) = \Q$. Since Theorem \ref{xx} shows that $\Q$ is the minimum distance of $\C$, the proof is complete. 
\end{proof}

%
%

\begin{remark}
\label{OtherPfeasyineq}
Except for the last line in the proof of above theorem, the fact that $\Q$ is the minimum distance of $\C$ has not been used anywhere. In fact, our proof of Theorem \ref{thmS1} shows that if $f\in\bigwedge^{m-\ell }V$ is Schubert decomposable, then $c_f$ is a codeword of $\C$ of weight $\Q$. Since it is easy to construct $(m-\ell)$-dimensional subspaces $W$ of $V$ such that $\dim W \cap A_{p_i} = \a_{p_i} - p_i$ for $i=1, \dots , u+1$, 
we can deduce that $d(\C) \le \Q$. This provides an alternative proof of Proposition \ref{easyineq}.
\end{remark}

We will now prove that the converse of Theorem \ref{thmS1} is true provided that the element 
 $f$ of $\bigwedge^{m-\ell }V$ is assumed to be decomposable. 

\begin{theorem}
\label{thmS2}
Assume that $f\in\bigwedge^{m-\ell }V$ is decomposable. If $c_f$ is a minimum weight codeword of $\C$, 
then $f$ is Schubert decomposable.
\end{theorem}

\begin{proof}

We use induction on $\ell$. If $\ell =1$, then $u=0$ and there is nothing to prove. Suppose $\ell > 1$ and that the result holds for values of $\ell$ smaller than the given one. Let $f\in\bigwedge^{m-\ell }V$ be a  decomposable element such that $c_f$ is a minimum weight codeword of $\C$. In particular, $c_f \ne 0$ and if 
 $E$ and $F$ denote the subsets of $A_{\ell}$ associated to $f$, as in Section \ref{sec3}, 
then $t:= \codim_{A_{\ell}} E$ satisfies $1\le t  \le \a_{\ell}$. 
 Moreover, by Theorem \ref{xx}, $c_{f\wedge x}$ is a minimum  weight codeword of $\Cp$ for every $x\in F$ and furthermore, we either have $t=1$ or we have $t > 1$ and $\a_\ell - \a_{\ell-1}=1$.
 
Note that for an arbitrary $x\in F$, by induction hypothesis we see that $g:=f\wedge x$ is Schubert decomposable. Hence $\dim V_g \cap A_{p_i} = \a_{p_i} - p_i$ for all $i=1, \dots , u$. Now since $f$ is decomposable, 
 so is $g$; moreover, $V_f$ is a hyperplane in $V_g$. Hence 
 \begin{equation}
 \label{pqi}
 q_i:= \dim V_f \cap A_{p_i} = \a_{p_i} - p'_i \quad \text{where} \quad p'_i = p_i \text{ or } p_i +1 \; \text{ for } i=1, \dots , u.
 \end{equation}
Let us also note that it suffices to show that $V_g\cap A_{p_u} \subseteq V_f\cap A_{p_u}$ 
because in that case, we obtain $V_{f}\cap A_{p_{u}} = V_{g}\cap A_{p_{u}}$, the other inclusion 
being trivial, and consequently, 
$\dim V_{f}\cap A_{p_{i}} = {\alpha }_{p_{i}} - p_{i}$  for all 
$i=1, \dots , u$, i.e., $f$ is Schubert decomposable.
We will now divide the proof into two cases according as $t=1$ or $t > 1$. 
%

\medskip
\noindent
{\bf Case 1.} $t =1$

In this case $p_u \le \ell -1$ and by Lemma \ref{F}, $A_{p_u} \subseteq E$. Let $x$ be an arbitrary element of 
$F$ and as before, let $g = f\wedge x$. 
Let $y\in V_g\cap A_{p_u}$. Since $y \in V_g = V_f + \langle x \rangle$, we can write $y = z + \lambda x$ for some $z\in V_f$ and $\lambda \in \Fq$. Also since $y \in A_{p_u}$, we find $y\in E$
 and so $c_{f\wedge y}=0$. Moreover, $c_{f\wedge z}=0$ simply because $z\in V_f$. It follows that 
$$
0 = c_{f\wedge y} = c_{f\wedge z} +  \lambda \, c_{f\wedge x} =  \lambda \, c_{f\wedge x} \quad \text{and hence} \quad \lambda = 0 \; \text{ so that } \; y \in V_f\cap A_{p_u}.
$$ 
Thus 
$V_g\cap A_{p_u}  \subseteq V_f\cap A_{p_u}$ and so, as noted before, $f$ is Schubert decomposable.
 
\medskip
\noindent
{\bf Case 2.} $t >1$

In this case $\a_\ell - \a_{\ell-1}=1$ and so $p_u < \ell -1$, i.e.,  
$p_u+2 \le \ell$. Recall that as per our convention $p_{u+1}:= \ell$ and so in view of \eqref{pqi} and Lemma \ref{S}, we set $p'_{u+1}:= \ell$ and $q_{u+1}:= \a_{\ell} - \ell$. Now we can recursively find $y_1, \dots y_{m-\ell}$ such that 
$$
 V_f \cap A_{p_i} = \langle y_1, \dots , y_{q_i} \rangle  \text{ for } i=1, \dots , u+1 \quad \text{and} \quad 
 V_f =  \langle y_1, \dots , y_{m- \ell} \rangle.
$$
In view of the dimension formulas \eqref{pqi} (that are also valid for $i=u+1$), this ensures that no nontrivial linear combination of $y_{q_i+1}, \dots , y_{m-\ell}$ is in $A_{p_i}$ for each $i=1, \dots , u+1$. 
By recursively extending these bases of $ V_f \cap A_{p_i}$ to $A_{p_i}$, we can also find $z_1, \dots , z_{\ell} \in A_{\ell}$ such that 
$$
 A_{p_i} = \langle y_1, \dots , y_{q_i},  z_1, \dots , z_{p'_i} \rangle  \quad \text{for } i=1, \dots , u+1. 
$$
Now consider the 
subspaces $L$ and $L'$ of $A_{\ell}$ defined by
$$
L = \langle z_1, \dots , z_{\ell} \rangle \quad \text{and} \quad L' = \langle z_1, \dots , z_{p'_u}, z_{p'_u+2}, \dots , z_{\ell} \rangle.
$$
It is clear that $\dim L = \ell = \dim L' +1$ and also that 
$$
\dim (L \cap  A_{p_i} )= \dim (L'  \cap A_{p_i})  = \dim \langle   z_1, \dots , z_{p'_i} \rangle  = p'_i \, \ge \, p_i \; \text{ for } i=1, \dots , u.
$$
Hence,  in view of \eqref{eq:7}, we see that $L\in \O$ and $L'\in \Op$. Moreover, by our choice of $z_1, \dots , z_{\ell}$, it is clear that $V_f \cap L = \{0\}$. Since $f$ is decomposable, it follows from \eqref{vanishing} that $c_f(L)\ne 0$. Hence if we let $x:= z_{p'_u+1}$, then we find $c_{f\wedge x}(L')\ne 0$. It follows that $x\in F$ 
and so induction hypothesis applies to $g = f\wedge x$ for this choice of $x$. Thus, $g$ is Schubert decomposable. Moreover, if $y\in V_g\cap A_{p_u}$, then being an element of $V_g = V_f + \langle x \rangle$, we can write 
$$
y = z + \lambda x \text{ for some } z\in \langle y_1, \dots , y_{m- \ell} \rangle  \text{ and } \lambda \in \Fq.
$$
On the other hand, being an element of $A_{p_u}$, we see that 
$$
y \in \langle y_1, \dots , y_{q_u},  z_1, \dots , z_{p'_u} \rangle .
$$
Thus, if $\lambda \ne 0$, then $x = z_{p'_u+1}$ can be expressed as a linear combination of $y_1, \dots , y_{m-\ell},  z_1, \dots , z_{p'_u}$, which contradicts the choice of $y$'s and $z$'s. So $\lambda = 0$ and $y \in V_f\cap A_{p_u}$. Thus 
$V_g\cap A_{p_u}  \subseteq V_f\cap A_{p_u}$ and so 
$f$ is Schubert decomposable. 
\end{proof}

In view of Theorems \ref{thmS1} and \ref{thmS2}, we make the following conjecture.

\begin{conjecture}
\label{conj}
\label{DandMWCisSD}
Minimum weight codewords of the Schubert code $\C$ are precisely the codewords corresponding to Schubert decomposable elements of $\bigwedge^{m-\ell}V$. 
\end{conjecture}

\section{Completely Non-consecutive Case}
\label{newsec5}

We will continue to use the notation and terminology of the last three sections. 
The main result of this section is an affirmative answer to Conjecture \ref{conj} when 
the dimension sequence $\a = (\a_1, \dots , \a_{\ell})$ of the fixed partial flag 
$A_1 \subset \dots \subset A_{\ell}$ is  completely non-consecutive, i.e., when $\a_i - \a_{i-1} \ge 2$ for  $1 < i \le \ell$.


\begin{theorem}
\label{thmCNC}
Assume that $\a$ is  completely non-consecutive. 
If $c$ is  a minimum weight codeword of $\C$, then $c = c_h$ for some decomposable $h\in\bigwedge^{m-\ell}V$. 
 \end{theorem}
 
 \begin{proof}
 We will use induction on $\ell$. The result clearly holds when $\ell=1$ since every nonzero element of 
 $\bigwedge^{m-1}V$ is decomposable. Suppose $\ell > 1$ and the result is true for values of $\ell$ smaller than the given one. Let $c$ be a minimum weight codeword of $\C$. Fix $f\in \bigwedge^{m-\ell}V$ such that $c=c_f$, and let $E, F$ be as in Section \ref{sec3}. Since $\a_\ell- \a_{\ell-1} \ge 2$, by Theorem \ref{xx} we see that $\codim_{A_\ell}E = 1$ and $f\wedge x$ is a minimum weight codeword of $\Cp$ for every $x\in F$. Moreover,  $A_{\ell -1} \subseteq E$, thanks to Lemma~\ref{x}. Thus, we can and will choose a basis $\{e_1, \dots , e_m\}$ of $V$ such that 
\begin{equation}
\label{gradedbasis}
 A_i = \langle e_1, \dots , e_{\a_i}\rangle \; \text{ for } \; i=1, \dots, \ell \quad \text{and} \quad E = 
  \langle e_1, \dots , e_{\a_\ell - 1}\rangle.
\end{equation}
Let $x:= e_{\a_\ell}$. Clearly, $x\in F$ and hence $f\wedge x$ is a minimum weight codeword of $\Cp$. Moreover, $\a'$ is completely non-consecutive. So by induction hypothesis, $c_{f\wedge x} = c_g$ for some decomposable $g\in \bigwedge^{m-\ell + 1}V$. Moreover, by Theorem \ref{thmS2}, $g$ is Schubert decomposable, and so $\dim V_g \cap A_i = \a_i -i$ for $i=1, \dots , \ell -1$. Thus, we can recursively choose 
$z_1, \dots , z_{\ell-1}$ such that 
$$
z_1 \in A_1 \setminus (V_g \cap A_1) \text{ and } z_i \in A_i \setminus \left( \langle z_1, \dots , z_{i-1} \rangle +  V_g \cap A_i \right) \text{ for } \; i=2, \dots, \ell -1. 
$$
In particular, $z_1, \dots , z_{\ell - 1}$ span an $(\ell-1)$-dimensional subspace, say $B_{\ell-1}$ of $A_{\ell-1}$  such that $A_{\ell-1}= B_{\ell-1} + \left( V_g \cap A_{\ell-1}) \right)$. This implies that $V_g \cap B_{\ell-1} = \{0\}$. Also since $\dim V_g = m -\ell +1$ and $x\not\in A_{\ell -1}$, 
we see that $\dim V_g \cap ( B_{\ell-1} + \langle x \rangle ) \ge 1$. Hence, $V_g$ contains an element of the form $b+x$ for some $b \in B_{\ell -1}$. Consequently, we can find 
$g_1, \dots g_{m-\ell} \in V$ such that $g_1, \dots g_{\a_{\ell-1} -(\ell -1)}$ span  $V_g \cap A_{\ell-1}$ and 
$$
g = g_1 \wedge \dots \wedge g_{m-\ell} \wedge (b+x)  = g'\wedge b + g'\wedge x, \quad \text{where } \; 
g':= g_1 \wedge \dots \wedge g_{m-\ell}.
$$
Note that $V_g \cap A_{\ell-1} = \langle g_1, \dots g_{\a_{\ell-1} -(\ell -1)} \rangle \subseteq V_{g'} \cap A_{\ell-1}
\subseteq V_g \cap A_{\ell-1}$. Thus,  
$$
V_{g'} \cap A_{\ell-1} = V_g \cap A_{\ell-1} \quad \text{and} 
\quad \dim ( V_{g'} \cap A_{\ell-1} ) = \a_{\ell-1} - (\ell-1). 
$$
We claim that $c_{g'\wedge b } = 0$. This is clear if $b\in V_{g'}$. Now suppose $b\not\in V_{g'}$. Then 
$$
 V_{g' \wedge b} \cap A_{\ell-1} 
=  ( V_{g'} + \langle b \rangle) \cap A_{\ell-1} 
=  ( V_{g'} \cap A_{\ell-1} ) + \langle b \rangle
\text{ has dimension } \a_{\ell-1} - \ell + 2
$$
and therefore $V_{g' \wedge b} \cap A_{\ell-1} $ has 
nonzero intersection with any $(\ell-1)$-dimensional subspace of $A_{\ell-1}$. In particular, $ U \cap V_{g' \wedge b}  \ne \{0\}$ for every $U\in \Op$. Thus, in view of \eqref{vanishing}, the claim is proved. From the claim, it follows that $c_g = c_{g'\wedge x}$. Writing each of $g_1, \dots g_{m-\ell}$ as a linear combination of $e_1, \dots , e_m$ and noting that $x= e_{\alpha_{\ell}}$, we see that 
$g'\wedge x = h \wedge x$, where $h $ is a decomposable element of $\bigwedge^{m-\ell}V$ of the form 
$h_1 \wedge \dots \wedge h_{m-\ell}$, where each of $h_1, \dots, h_{m-\ell}$ is in the $(m-1)$-dimensional space $V$ spanned by $\{e_1, \dots , e_m\}\setminus\{x\}$. 
We will now proceed to prove that $c_f = c_h$. 

Let $L \in \O$ and let $P= u_{\ell} \wedge \dots \wedge u_1$ with $u_i \in A_i$ for $1\le i \le \ell$, be the  representative of $L$ in $\bigwedge^{\ell}V$ among 
the fixed representatives $P_1, \dots , P_{n_\a}$ as in \S \ref{subsec:2.3}. We wish to show that 
$c_f(P):= f \wedge u_{\ell}\wedge \dots \wedge u_1$ is equal to $c_h(P): = h \wedge u_{\ell}\wedge \dots \wedge u_1$. Since $c_{f\wedge x} = c_{h \wedge x}$ , we readily see that $f \wedge x \wedge u_{\ell -1}\wedge \dots \wedge u_{1} = h \wedge x \wedge u_{\ell -1}\wedge \dots \wedge u_{1}$. We will now 
consider two cases. 
First, suppose $u_{\ell} \in E$. Then $c_{f \wedge u_{\ell} }$ is the zero codeword in $\Cp$ and hence $c_f(P)=0$. On the other hand, 
by \eqref{gradedbasis} and our choice of $h$, we see that  $V_h + E$ is a subspace of $V'$. Since $\dim V_h + \dim L = m > \dim V'$, we must have $V_h \cap L \ne \{0\}$ and so by \eqref{vanishing}, we obtain $c_h(L) =0$ as well. Now suppose $u_{\ell} \not\in E$. Then $u_{\ell} = v_{\ell} + \lambda x$ for a unique $v_{\ell} \in E$ and $\lambda \in \Fq$ with $\lambda \ne 0$. As in the previous case,
$f \wedge v_{\ell} \wedge u_{\ell-1} \wedge \dots \wedge u_{1}  = 0 = h \wedge v_{\ell} \wedge u_{\ell-1} \wedge \dots \wedge u_{1} $. Consequently, 
$$
c_h(P) = \lambda \left( h \wedge x \wedge u_{\ell -1}\wedge \dots \wedge u_{1}  \right) 
= \lambda (f \wedge x \wedge u_{\ell -1}\wedge \dots \wedge u_{1}) 
= f \wedge u_{\ell} \wedge u_{\ell-1} \wedge \dots \wedge u_{1},
$$
and thus $ c_h(P) =  c_{f} (P)$. This establishes $c_f = c_h$ and so the theorem is proved. 
 \end{proof}	
 
 As an immediate consequence of the above theorem, we see that Conjecture \ref{DandMWCisSD} holds in the affirmative when $\a$ is  completely non-consecutive. 
 
\begin{corollary}
\label{CorCNC}
Assume that $\a$ is  completely non-consecutive. Then the minimum weight codewords of 
$\C$ are precisely the codewords corresponding to Schubert decomposable elements of $\bigwedge^{m-\ell}V$. More precisely, for any $c\in \C$, 
$$
c \text{ has minimum weight } \Longleftrightarrow c = c_h \text{ for some Schubert decomposable } h \in \bigwedge^{m-\ell}V.
$$
 \end{corollary}  
 
 \begin{proof}
 Follows from Theorems \ref{thmS1}, \ref{thmS2} and \ref{thmCNC}.
 \end{proof}
 
 
We note that a special case of the last corollary implies  that Conjecture \ref{DandMWCisSD} holds in the affirmative when $\ell=2$. 
 
 \begin{corollary}
\label{CorCNCTwo}
The minimum  weight codewords of $C_{\a}(2,m)$ are precisely the codewords corresponding to Schubert decomposable elements of $\bigwedge^{m-2}V$.
 \end{corollary}  
 
 \begin{proof}
Follows from Corollary  \ref{CorCNC} and the characterization by Nogin \cite{N} of minimum weight codewords of Grassmann codes because when $\ell=2$, the pair $\a$ must be either completely consecutive or completely nonconsecutive.
 \end{proof}
 
 \section{Enumeration and Generation}
 
In this section, we consider the problem of enumerating the number of minimum weight codewords of Schubert codes and also of determining whether or not the minimum weight codewords  generate the Schubert code $\C$. The case of Grassmann codes, 
which is when $\a_i = m-\ell +i$ for $i=1, \dots , \ell$, is well-known. Here we know that the number of minimum weight codewords is 
$
(q-1) {{m}\brack {\ell}}_q
$
and also that the minimum weight codewords of the Grassmann code $C(\ell,m)$ generate $C(\ell,m)$. 
Both these assertions follow readily from Nogin's characterization of the minimum weight codewords as those that correspond to decomposable elements of $\bigwedge^{m-\ell}V$. In the case of Schubert codes, we have noted earlier that the 
map given by $f\mapsto c_f$ from $\bigwedge^{m-\ell}V$ onto 
$\C$ is ``many-to-one''. But for studying the minimum weight codewords of $\C$, it suffices to consider the restriction of this map to the set of Schubert decomposable elements of $\bigwedge^{m-\ell}V$, and examine 
to what extent it is injective. This is done in the next two lemmas.

\begin{lemma}
\label{lemE1}
Let $f, g\in \bigwedge^{m-\ell}V$ be Schubert decomposable elements such that $c_f = c_g$. Then 
$V_f \cap A_{p_i} = V_g \cap A_{p_i}$ for all $i=1, \dots , u+1$. 
\end{lemma}

\begin{proof}
Assume the contrary, i.e.,  suppose $V_f \cap A_{p_i} \ne V_g \cap A_{p_i}$ for some $i \le u+1$. We will  assume that $i$ is the least positive integer with this property. Then 
\begin{equation}
\label{mini}
V_f \cap A_{p_j} = V_g \cap A_{p_j } \quad \text{for } 0 \le j < i, \quad\text{where} \quad  p_0=0 \text{ and }A_0:= \{0\} .
\end{equation}
Since $\dim (V_f \cap A_{p_i} )= \dim (V_g \cap A_{p_i})$, we see that 
$V_f \cap A_{p_i} \not\subseteq V_g \cap A_{p_i}$.
So there is some $x\in V_f \cap A_{p_i}$ such that $x \not\in V_g \cap A_{p_i}$. 
By \eqref{mini}, $x\not\in A_{p_j}$ for $0\le j < i$. Since $f,g$ 
are Schubert decomposable, we can recursively choose $x_1, \dots , x_{\ell} \in A_{\ell}$ such that 
$$
(V_g \cap A_{p_j} )+ \langle x_1, \dots , x_{p_j}\rangle = A_{p_j} \quad \text{for } j=1, \dots , u+1 \quad \text{and} \quad x_{p_i} = x,
$$
where $p_{u+1}=\ell$. 
Now let $L= \langle x_1, \dots , x_{\ell}\rangle$. By our choice of $x_1, \dots , x_{\ell}$, it is clear from \eqref{eq:7} that $L\in \O$ and $V_g \cap L = \{0\}$. Since $g$ is decomposable, this implies $c_g(L)\ne 0$. 
On the other hand, since $x\in V_f \cap L$, we see from \eqref{vanishing} that $c_f(L)=0$. Thus, $c_f\ne c_g$, which contradicts the hypothesis. 
\end{proof}

A partial converse of the above result is also true. 

\begin{lemma}
\label{lemE2}
Let $f, g\in \bigwedge^{m-\ell}V$ be Schubert decomposable elements such that $V_f \cap A_{p_i} = V_g \cap A_{p_i}$ for all $i=1, \dots , u+1$. Then $c_f = \lambda \, c_g$ for some $\lambda\in \Fq\setminus\{0\}$.
\end{lemma}

\begin{proof}
Let $r_i: = \a_{p_i} -p_i$ for $1\le i \le u+1$ and $r:=r_{u+1}$. We can recursively find linearly independent 
$f_1, \dots , f_r, \, e_1, \dots , e_{\ell} \in A_{\ell}$ such that for each $i=1, \dots , u+1$, 
$$
V_f \cap A_{p_i} = V_g \cap A_{p_i} = \langle f_1, \dots , f_{r_i}\rangle \quad \text{and} \quad
A_{p_i} = (V_f \cap A_{p_i})  + \langle e_1, \dots , e_{p_i}\rangle. 
$$
Extend $\{f_1, \dots , f_r\}$ to bases $\{f_1, \dots , f_{m-\ell} \}$ and $\{f_1, \dots , f_r, g_1, \dots , g_{m - \a_\ell} \}$ of $V_f$ and $V_g$ respectively, such that $f= f_1 \wedge \dots \wedge f_{m-\ell}$ and 
$g =  f_1 \wedge \dots \wedge f_r \wedge g_1 \wedge \dots \wedge g_{m - \a_\ell}$. Note that $V_f \cap \langle e_1, \dots , e_{\ell}\rangle = \{0\}$ 
and thus  $\{f_1, \dots , f_{m-\ell}, e_1, \dots , e_{\ell} \}$
is a basis of~$V$.  In particular, for each $j=1, \dots , m- \a_\ell$, we can write $g_j = x_j + y_j +z_j$ for unique $x_j \in \langle f_1, \dots , f_r\rangle$, $y_j \in \langle f_{r+1}, \dots , f_{m-\ell}\rangle$ and $z_j \in \langle e_1, \dots , e_{\ell}\rangle$. Hence, by multilinearity, we see that $g$ is a finite sum of elements of the form 
$$
h = f_1 \wedge \dots \wedge f_r \wedge h_1 \wedge \dots \wedge h_{m - \a_\ell}, \text{ where } \; 
h_j \in \{x_j, y_j, z_j\} \text{ for } j=1, \dots , m-\a_\ell.
$$ 
Now if $h_j = x_j$ for some $j$, then clearly $h=0$. Also, if $h_j = z_j$ for some $j$, then we find that $h$ is a decomposable element of $ \bigwedge^{m-\ell}V$ such that 
$$
\dim (V_h \cap A_{\ell}) \ge \dim \langle f_1, \dots , f_{r}, z_i \rangle = \a_{\ell} - \ell +1
$$ 
and thus $c_h=0$, because otherwise Lemma \ref{S} is contradicted. It follows that $c_g = c_{h^*}$, where 
$h^* := f_1 \wedge \dots \wedge f_r \wedge y_1 \wedge \dots \wedge y_{m - \a_\ell} $. 
By Lemma \ref{lem:SchubNonzero}, $c_g \ne 0$ and hence $h^*\ne 0$. Consequently, $y_1 ,\dots , y_{m - \a_\ell}$
are linearly independent elements of $\langle f_{r+1}, \dots , f_{m-\ell}\rangle$, and therefore 
$y_1 \wedge \dots \wedge y_{m - \a_\ell}$ and $f_{r+1} \wedge \dots \wedge f_{m-\ell}$ differ by a nonzero scalar. This implies that $c_f = \lambda \, c_g$  
for some $\lambda\in \Fq\setminus\{0\}$. 
\end{proof}

\begin{remark}
It may be noted that with hypothesis as in Lemma \ref{lemE2}, the stronger conclusion that $f = \lambda \, g$ for some $\lambda\in \Fq\setminus\{0\}$ or equivalently, $V_f=V_g$, is not true, in general. Indeed, this is indicated by the proof and examples are easy to construct. For instance, if $\ell = 2$, $m=4$ and $\alpha = (2,4)$, then $f= e_1 \wedge e_3$ and $g = e_1 \wedge (e_2 + e_3)$ are Schubert decomposable elements of $\bigwedge^2V$ such that 
$c_f = c_g$, but $f$ and $g$ do not differ by a scalar.  Here $e_1, e_2, e_3, e_4$ denote the elements of a fixed basis of $V$. 
\end{remark}

The following lemma is a variant of \cite[Lem. 3]{GT}, but with a simpler formula and a more direct proof. 

\begin{lemma}
\label{subspaces}
Let $B$ be a finite-dimensional vector space over $\Fq$ and let $A$ be a subspace of $B$ and $R$ a subspace of $A$. 
Suppose $b = \dim B$, $a= \dim A$ and $r = \dim R$. Let $u$ be any 
integer with $0 \le u\le b$ and, as before, let $G_u(B)$ denote the Grassmannian of $u$-dimensional subspaces of $B$. Then 
$$
\left| \left\{ U\in G_u(B) : U \cap A = R \right\} \right| = {{b-a}\brack{u-r}}_q q^{(a - r)(u - r)}.
$$
\end{lemma}

\begin{proof}    
Let $\mathfrak{U}:= \left\{ U\in G_u(B) : U \cap A = R \right\}$
We have a natural surjective map 
$$
\psi: \mathfrak{U} \to G_{u-r}(B/A) \quad \text{given by}\quad 
U \longmapsto \frac{U+A}{A} \simeq \frac{U}{U\cap A} = \frac{U}{R}.
$$
Note that 
an arbitrary element of $G_{u-r}(B/A)$ is of the form $T/A$, where $T$ is a subspace of $B$ containing $A$ with $\dim T = a+ u-r$. 
Fix such $T/A$. Then
$$Ne
\psi^{-1}\left( T/A \right) = \left\{ U\in G_u(B) : U \cap A = R \text{ and } A + U = T\right\}
$$
To estimate the cardinality of this fibre, let us fix an ordered basis $\{x_1, \dots , x_a\}$ of $A$ such that $\{x_1,  \dots , x_r\}$ is a basis of $R$. 
Now $T$ has an ordered basis of the form $\{x_1, \dots , x_a, y_1, \dots , y_{u-r}\}$, and $U:= \langle x_1, \dots , x_r,  y_1, \dots , y_{u-r}\rangle$ is in $\psi^{-1}\left( T/A \right)$. Moreover, every element of $\psi^{-1}\left( T/A \right)$ can be obtained in this manner 
by choosing 
$z_1, \dots , z_{u-r} \in T$ such that $\langle z_1, \dots , z_{u-r} \rangle \cap A = \{0\}$ and $z_1, \dots , z_{u-r}$ are linearly independent.  Since $|T| = q^{a+u-r}$ and $|A| = q^a$, 
the number of ordered $(u-r)$-tuples 
$(z_1, \dots , z_{u-r})$ with this property is
\begin{equation} 
\label{count1}
\left(q^{a+u-r} - q^{a} \right) \left(q^{a+u-r} - q^{a+1} \right) \cdots \left(q^{a+u-r} - q^{a+u-r-1} \right)
\end{equation}
Two ordered $(u-r)$-tuples $(y_1, \dots , y_{u-r})$ and  $(z_1, \dots , z_{u-r})$ 
give rise to the same subspace 
if and only if 
$$
\begin{bmatrix} z_1 \\ \vdots \\ z_{u-r} \end{bmatrix} = 
\begin{bmatrix} & & & \vdots & & & & \\ & & & \vdots & & & & \\ & P & & \vdots &  & Q &  & \\  & & & \vdots & & & & \\& & & \vdots & & & &  \end{bmatrix} 
\begin{bmatrix} x_1 \\ \vdots \\ x_r \\ y_1 \\ \vdots \\ y_{u-r} \end{bmatrix} 
$$
for some $(u-r)\times r$ matrix $P$ and $(u-r)\times (u-r)$ nonsingular matrix $Q$ with entries in $\Fq$. Indeed, in that case the two ordered bases $\{x_1, \dots , x_r,  y_1, \dots , y_{u-r}\}$ and $\{x_1, \dots , x_r,  z_1, \dots , z_{u-r}\}$ will differ by a nonsingular $u \times u$ matrix of the form 
$$
\begin{bmatrix} I_r  &\! \vdots & \! \mathbf{0} \\  \dots & \! \dots & \! \dots \\ P &\! \vdots &\! Q  
\end{bmatrix}
$$
where $I_r$ denotes the identity matrix of size $r\times r$ and $\mathbf{0}$ denotes the $r\times (u-r)$ matrix all of whose entries are zero. The number of ways in which matrices $P$ and $Q$ of the kind above can be chosen is clearly given by 
\begin{equation} 
\label{count2}
 q^{r(u-r)} \left( q^{u-r} -1 \right) \left( q^{u-r} -q \right) \cdots  \left( q^{u-r} - q^{u-r-1} \right). 
\end{equation}
It follows that the cardinality of $\psi^{-1}\left( T/A \right)$ is obtained by dividing the expression in \eqref{count1} by that in \eqref{count2}. Thus,  
$$
|\psi^{-1}\left( T/A \right)| = q^{(a-r)(u-r)} \quad \text{and hence } \quad |\mathfrak{U}| = {{b-a}\brack{u-r}}_q q^{(a - r)(u - r)}.
$$
This completes the proof. 
\end{proof}

We are now ready to prove the main result of this section. 

\begin{theorem}
\label{thmE}
The number of codewords of $\C$ corresponding to Schubert decomposable elements of $ \bigwedge^{m-\ell}V$ is equal to 
$$
M_{\a} : = (q-1) q^{\mathsf{P}} \prod_{j=0}^{u}  { {\a_{p_{j+1}} - \a_{p_j} } \brack{p_{j+1} - p_j}}_q 
$$
where, as per our usual conventions, $p_0=0$, $p_{u+1} = \ell$, and $\a_0 = 0$, and where
$$
\mathsf{P} = \sum_{j=1}^u p_j \left(\a_{p_{j+1}} - \a_{p_j} - p_{j+1} + p_j \right). 
$$
Consequently, the number of minimum weight codewords of $\C$ is at least $M_{\a}$. Moreover, if $\a$ is completely non-consecutive, then the number of minimum weight codewords of $\C$ is exactly $M_{\a}$.
\end{theorem}

\begin{proof}
Let us temporarily denote by $S_{\a}$ the set that we wish to enumerate, i.e., let 
$$
S_{\a} : = \{c_f : f \in \bigwedge^{m-\ell}V \text{ is Schubert decomposable} \}.
$$ 
Note that by Theorem \ref{thmS1}, elements of $S_{\a}$ are minimum weight codewords of $\C$ and in particular, nonzero elements of $\Fq^{n_{\a} }$. 
Consider the 
relation $\sim$ on $S_{\a}$ defined by $c \sim c' \Leftrightarrow c = \lambda c'$ for some 
$\lambda\in \Fq\setminus\{0\}$. Clearly, $\sim$ is an equivalence relation and each equivalence class has cardinality $(q-1)$. Thus, if we denote 
by $\mathfrak{S}_{\a}$ the set of all equivalence classes, then $|S_{\a} | = (q-1)| \mathfrak{S}_{\a}|$. On the other hand, there is a similar equivalence relation (viz., proportionality) on the set of all decomposable elements of $\bigwedge^{m-\ell}V$, and the map $f\mapsto V_f$ sending a decomposable element to its annihilator clearly gives a bijection between the set of equivalence classes and the Grassmannian $G_{m-\ell}(V)$ of $(m-\ell)$-dimensional subspaces of $V$. 
This equivalence relation 
preserves Schubert decomposability and the set of equivalence classes of 
Schubert decomposable elements of $ \bigwedge^{m-\ell}V$ is clearly in bijection with 
$$
\Lambda_{\a}: = \left\{W \in G_{m-\ell}(V): \dim W \cap A_{p_i} = \a_{p_i} - p_i \text{ for all } i=1, \dots , u+1\right\}.
$$
For any $c\in S_{\a}$, denote by $[c]$ its equivalence class in $\mathfrak{S}_{\a}$. Then the map
$$
\theta : \Lambda_{\a} \to \mathfrak{S}_{\a} \quad \text{given by} \quad \langle w_1, \dots , w_{m-\ell} \rangle \mapsto \left[ c_{w_1 \wedge \dots \wedge w_{m-\ell}} \right]
$$
is clearly well-defined and surjective. 
By Lemmas \ref{lemE1} and \ref{lemE2}, for any $W, W' \in \Lambda_{\a}$, 
$$
\theta(W) = \theta(W') \Leftrightarrow  W \cap A_{p_i} = W' \cap A_{p_i} \text{ for all } i=1, \dots , u+1 
\Leftrightarrow  W \cap A_{\ell} = W' \cap A_{\ell}.
$$
It follows that for any $[c]\in \mathfrak{S}_{\a}$, the fibre $\theta^{-1}([c])$ is in bijection with the set of all $W\in G_{m-\ell}(V)$ such that $W\cap A_{\ell}$ is equal to a fixed $(\a_{\ell} - \ell)$-dimensional subspace, say $W_{\ell}$, of $A_{\ell}$. Hence, by Lemma \ref{subspaces}, 
\begin{equation}
\label{Rel1}
| \Lambda_{\a}| = \! \sum_{[c]\in \mathfrak{S}_{\a} } \left|\theta^{-1}([c]) \right| = {{m- \alpha_\ell}\brack{m - \ell - (\alpha_\ell - \ell)}}_q 
q^{\ell(m - \alpha_\ell)} \left|\mathfrak{S}_{\a} \right| =  
q^{\ell(m - \alpha_\ell)}
\left|\mathfrak{S}_{\a} \right|.
\end{equation}
Now  let $r_i: = \a_{p_i} -p_i$ for $1\le i \le u+1$ and consider the following sequence of maps 
$$
\Lambda_{\a} \stackrel{\pi_{u+1}}{\longrightarrow}  \Lambda_{u+1} \stackrel{\pi_{u}}{\longrightarrow}  \Lambda_{u} \stackrel{\pi_{u-1}}{\longrightarrow} \;  \dots \; 
\stackrel{\pi_{2}}{\longrightarrow}  \Lambda_{2} \stackrel{\pi_{1}}{\longrightarrow}  \Lambda_{1} 
$$
where for $1\le j \le u+1$,  the set $\Lambda_j $ is defined by 
$$
\Lambda_j : = \left\{U \in G_{r_j}(A_{p_j}) :  \dim U\cap A_{p_i} = \a_{p_i} - p_i \text{ for } 1\le i < j \right\},
$$
while for $1\le j \le u+1$, 
the map $\pi_{j} : \Lambda_{j+1} \to \Lambda_j $ is defined by $\pi_j(U) = U \cap A_{p_{j}}$ for any $U \in \Lambda_{j+1}$, where, by convention, we have set
$$
\Lambda_{u+2}: = \Lambda_{\a}, \; \; r_{u+2} := m - \ell, \; \; A_{p_{u+2}}: = V,\; \;  \a_{p_{u+2}}: = m, \;  \text{and as before, } \; 
p_{u+1}: = \ell.
$$
By Lemma \ref{subspaces}, we see that the cardinality $N_j: = | \pi_j^{-1}(U)|$ of the fibre of any $U\in \Lambda_j$ is independent of the choice of $U$ and is given by 
\begin{equation}
\label{Rel2}
N_j  
=  {{\a_{p_{j+1}} - \a_{p_j}}\brack{r_{j+1} - r_j}}_q  q^{p_j (r_{j+1} - r_j)} 
\quad \text{for } j=1, \dots , u+1. 
\end{equation}
It follows that 
\begin{equation}
\label{Rel3}
| \Lambda_{\a}| = N_{u+1} | \Lambda_{u+1}| = N_{u+1}N_u | \Lambda_{u}| = \dots   = N_{u+1}N_u \cdots N_1  | \Lambda_{1}| .
\end{equation}
Now note that 
$$
N_{u+1} = {{m- \alpha_\ell}\brack{m - \ell - (\alpha_\ell - \ell)}}_q 
q^{\ell(m - \alpha_\ell)}  = q^{\ell(m - \alpha_\ell)}  \; \; \text{and} \; \;  
 | \Lambda_{1}|  = \left| G_{r_1}(A_{p_1})\right| =  {{\a_{p_{1}} }\brack{r_{1}}}_q .
$$
Substituting this and \eqref{Rel2} in \eqref{Rel3} and then comparing with \eqref{Rel1}, we obtain
$$
\left|\mathfrak{S}_{\a} \right| = \prod_{j=0}^u {{\a_{p_{j+1}} - \a_{p_j}}\brack{r_{j+1} - r_j}}_q  q^{p_j (r_{j+1} - r_j)} = \prod_{j=0}^u {{\a_{p_{j+1}} - \a_{p_j}}\brack{p_{j+1} - p_j}}_q  q^{p_j (r_{j+1} - r_j)} ,
$$
where, as before, we have set $p_0 = 0 = \alpha_{p_0} = r_0$. This implies that 
$|S_\a| = M_{\a}$. The remaining assertions 
follow from Theorem \ref{thmS1} and Corollary \ref{CorCNC}. 
\end{proof}

\begin{remark}
\label{RemEnum}
It is clear that if Conjecture \ref{DandMWCisSD} holds in the affirmative, then $M_{\a}$ given in Theorem \ref{thmE} is precisely the number of minimum weight codewords of $\C$. 
Note that when $\a$ is  completely consecutive, i.e., when $u=0$, 
we have $M_\a = (q-1) {{\a_\ell}\brack{\ell}}_q$, which is consistent with the result of Nogin \cite{N} mentioned earlier since in this case $\O$ is the Grassmannian $G_{\ell}(A_{\ell})$. 
\end{remark}

The question as to whether or not the minimum weight codewords of a code generate the code is often of some interest. It is a classical result that this is true in the case of binary Reed-Muller codes (see, e.g., \cite[Ch. 13,
§ 6]{MS}), whereas for $q$-ary generalized
Reed-Muller codes, it is not true, in general 
(see, e.g., \cite[Thm. 1]{DK}). For Grassmann codes as well as for related classes of codes such as affine Grassmann codes of an arbitrary level, the minimum weight codewords generate the code (see, e.g., \cite[Thm. 18 and Rem. 1]{BGH2}). However, we will show below that Schubert codes are, in general, not generated by their minimum weight codewords. 

\begin{theorem}
\label{minwtdont}
Assume that $\a$ has more than two consecutive blocks, i.e., $u>1$. Then the $\Fq$-linear subspace of $\C$ generated by the codewords corresponding to Schubert decomposable elements of $\bigwedge^{m-\ell}V$ 
is a proper subset of $\C$. 
\end{theorem}

\begin{proof}
As in the last proof, 
let $S_\a:= \{c_f : f\in \bigwedge^{m-\ell}V \text{ is Schubert decomposable}\}$. 
The hypothesis on $\a$ implies that $p_2+1 \le \ell$, and also that 
$$
\text{either (i) } \a_{p_1} \ge p_1+1 \quad \text{or} \quad  \text{ (ii) } \a_{p_1} = p_1\text{ and } \a_{p_2} \ge p_2+1.
$$
Now fix a basis $\{e_1, \dots , e_m\}$ of $V$ such that $A_i = \langle e_1, \dots e_{\a_i}\rangle$ for $i=1, \dots , \ell$. Also let $L= \langle e_1, \dots , e_{\ell}\rangle$ and $g:= e_{\ell+1} \wedge \dots \wedge e_m$. Then $g$ is a decomposable element of $ \bigwedge^{m-\ell}V$ such that $c_g(L)\ne 0$. 
Now suppose $f\in  \bigwedge^{m-\ell}V$ is any Schubert decomposable element.  Then 
in case (i) holds, i.e., when $\a_{p_1} \ge p_1+1$,  we find
 $$
 \dim V_f\cap A_{p_1} = \a_{p_1} - p_1 \ \quad \text{and} \quad  \dim L \cap A_{p_1} \ge p_1 +1
 $$
 and consequently, $\dim (V_f \cap L) \ge 1$, which in view of \eqref{vanishing} shows that $c_f(L)=0$.  On the other hand, if (ii) holds, then
  $$
 \dim V_f\cap A_{p_2} = \a_{p_2} - p_2 \ \quad \text{and} \quad  \dim L \cap A_{p_2} \ge p_2 +1
 $$
 and consequently, $\dim (V_f \cap L) \ge 1$, which implies once again that $c_f(L)=0$. It follows that if $c\in \C$ is any linear combination of elements of $S_\a$, then $c(L)\ne 0$. Hence, $c_g \in \C$ is not in the linear span of $S_\a$. 
%
\end{proof}

\begin{corollary}
Suppose $\a$ is completely non-consecutive and $\ell >2$. Then $\C$ is  not generated by its minimum weight codewords.
\end{corollary}

\begin{proof}
Since $\a$ is completely non-consecutive, we have $u = \ell -1$ and so $u > 1$. Thus the desired result follows from Corollary \ref{CorCNC} and Theorem \ref{minwtdont}. 
\end{proof}

\begin{remark}
\label{RemMinWt}
As in Remark \ref{RemEnum}, it is clear that if Conjecture \ref{DandMWCisSD} holds in the affirmative, then Theorem \ref{minwtdont} shows that $\C$ is not generated by its minimum weight codewords, provided $u>1$. In fact, our proof of Theorem \ref{minwtdont} shows that its conclusion as well as  the last assertion is also valid when $u=1$, provided $\a_{p_1} > p_1$.  On the other hand, when $u=0$, i.e., when $\a$ is consecutive, then $\O \simeq G_{\ell}(A_{\ell})$ and $\C$ is equivalent to the Grassmann code $C(\ell, \a_\ell)$. So we know from the work of Nogin \cite{N} that $\C$ is generated by its minimum weight codewords. Moreover, when $u=1$ and $\a_{p_1} = p_1$, then any $W\in \O$ satisfies $W\cap A_{p_1} = A_{p_1}$, i.e., $A_{p_1} \subseteq W$, and hence $W \mapsto W/A_{p_1}$ sets up a natural isomorphism between $\O$ and 
$G_{\ell - p_1}(A_{\ell}/A_{p_1})$. Consequently, 
$\C$ is equivalent to the Grassmann code $C(\ell - p_1, \a_\ell - \a_{p_1})$. So once again,  
 Nogin's result 
 implies that $\C$ is generated by its minimum weight codewords in this case. 
\end{remark}

\section*{Acknowledgments}

The first named author would like to thank Lucio Guerra for a preliminary discussion many years ago  about 
the minimum distance conjecture 
in the case 
$\ell=3$ and some correspondence related to \cite{X}. 
He is also grateful to Christian Krattenthaler for pointing out that the formula in \cite[Lem. 3]{GT} can be rewritten as a hypergeometric series and simplified using $q$-Chu-Vandermonde identity. Thanks are also due to the referees for their comments as well as suggestions for improving the exposition in an earlier version of this paper.


\begin{thebibliography}{AAAA}
 \bibitem{BGH2}
P. Beelen, S. R. Ghorpade, and T. H{\o}holdt,  {Duals of affine Grassmann codes and their relatives}, \emph{IEEE Trans. Inform. Theory}, \textbf{58} (2012), 
3843--3855.

\bibitem{BP} P. Beelen and F. Pi\~nero, The structure of dual Grassmann codes, \emph{Des. Codes Cryptogr.} {\bfseries 79} 
(2016), 451--470.
\bibitem{HC}
H. Chen, On the minimum distance of Schubert codes, 
{\em IEEE Trans. Inform. Theory} {\bfseries 46} (2000), 1535--1538.
 \bibitem{DK}
P. Ding and J. D. Key, Minimum-weight codewords as generators of generalized Reed-Muller codes, \emph{IEEE Trans. Inform. Theory} 
{\bfseries 46} (2000), 
2152--2158. 

\bibitem{GK}
S. R. Ghorpade and K. V. Kaipa, Automorphism groups of Grassmann codes, 
\emph{Finite Fields Appl.} {\bfseries 23} (2013),  80--102.
 	
\bibitem{GL}
S. R. Ghorpade and G. Lachaud,  Higher weights of Grassmann codes, 
\emph{Coding Theory, Cryptography and Related Areas} (Guanajuato, 1998),
 J. Buchmann, T. H{\o}holdt, H. Stichtenoth and H. Tapia-Recillas Eds.,
  Springer-Verlag, Berlin, 
  2000, pp. 122--131. 
 	\bibitem{GPP}
S. R. Ghorpade, A. R. Patil and H. K. Pillai,
{Decomposable subspaces,  linear sections of Grassmann varieties,
and higher weights of Grassmann codes}, 
{\em Finite Fields Appl.} {\bfseries 15} (2009),  54--68.
 	\bibitem{GT}
 S. R. Ghorpade and M. A. Tsfasman,
{\rm Schubert varieties, linear codes and enumerative combinatorics},
\emph{Finite Fields Appl.} {\bfseries 11}  (2005), 684--699.
 	
\bibitem{GV}
L. Guerra and R. Vincenti, On the linear codes arising from Schubert
varieties, {\em Des. Codes Cryptogr.} {\bfseries 33} (2004), 173--180.
 	
 	
\bibitem{HJR2}
J. P. Hansen, T. Johnsen, and K. Ranestad,
{\rm Grassmann codes and Schubert unions}, \emph{Arithmetic, Geometry and Coding Theory} (AGCT-2005, Luminy), F. Rodier and S. Vl{\u{a}}du{\c{t}} Eds., 
S\'emin. 
Congr.,  
vol.~21, Soc. Math. France,  Paris, 2010, pp. 103--121. 

\bibitem{M}
M. Marcus, \emph{Finite Dimensional Multilinear Algebra, Part II}, Marcel Dekker, New York, 1975.

\bibitem{MS}
F. J. MacWilliams and N. J. A. Sloane, \emph{The Theory of Error Correcting Codes}, Elsevier, New York, 1977.

\bibitem{N}
D. Yu. Nogin, Codes associated to Grassmannians,
\emph{Arithmetic, Geometry and Coding Theory} (Luminy, 1993), 
R. Pellikaan, M. Perret, S. G. Vl\u{a}du\c{t}, Eds.,
Walter de Gruyter, Berlin, 
 1996,  pp. 145--154.
 \bibitem{FP}
 F. Pi\~nero, The structure of dual Schubert union codes, 
\emph{IEEE Trans. Inform. Theory} 
{\bfseries 63} (2017), 
1425--1433.
\bibitem{CR1}
C. T. Ryan, An application of Grassmannian varieties to coding
theory, {\em Congr. Numer.} {\bfseries 57} (1987), 257--271.
\bibitem{CR2}
C. T. Ryan, Projective codes based on Grassmann varieties,
{\em Congr. Numer.} {\bfseries 57} (1987), 273--279. 
\bibitem{TVN}
M. Tsfasman, S. Vl{\u{a}}du{\c{t}} and D. Nogin, \emph{Algebraic Geometric Codes: Basic Notions}, Math. Surv. Monogr., vol. 139, Amer. Math. Soc., Providence, 2007.  
\bibitem{CKR}
C. T. Ryan and K. M. Ryan, The minimum weight of Grassmannian codes $C(k,n)$,
{\em Disc. Appl. Math.} {\bfseries 28} (1990), 149--156. 
 	
 	\bibitem{X} X. Xiang, {On The Minimum Distance Conjecture For Schubert Codes},  \emph{IEEE Trans. Inform. Theory} {\bfseries 54} (2008), 486--488. 
 	


  	 
 \end{thebibliography}
\end{document}